\documentclass[pra,twocolumn, superscriptaddress, preprintnumbers,amsmath,amssymb, nofootinbib]{revtex4-1}

\usepackage{hyperref}
\usepackage{xcolor}                     % To allow use of colour.
\usepackage{amsmath}
\usepackage{eufrak}
\usepackage{empheq}
\usepackage{amsthm}
\usepackage{graphicx}
\usepackage{braket}
\usepackage{mathtools}
\usepackage{graphicx}
\usepackage{amssymb}
\usepackage{bm}
\usepackage{theoremref}
\usepackage{enumerate}
\usepackage[capitalise]{cleveref}
\usepackage[T1]{fontenc} %for polish characters to work
\usepackage{lmodern}
\usepackage{tabularx, booktabs}
\newcolumntype{Y}{>{\centering\arraybackslash}X}
\input{Qcircuit}

\newtheorem{theorem}{Theorem}

\newtheorem{lemma}{Lemma}
\newtheorem{proposition}{Proposition}

\makeatletter

\newtheorem*{rep@theorem}{\rep@title}
\newcommand{\newreptheorem}[2]{%
\newenvironment{rep#1}[1]{%
 \def\rep@title{#2 \ref{##1}}%
 \begin{rep@theorem}}%
 {\end{rep@theorem}}}
\makeatother

\newreptheorem{theorem}{Theorem}

\newcommand{\ketbra}[2]{| \hspace{1pt} #1 \rangle \langle #2 \hspace{1pt} |}

\newcommand{\dya}[1]{\ket{#1}\!\bra{#1}}
\newcommand{\dyad}[2]{\ket{#1}\!\bra{#2}}        %dyad
\newcommand{\ip}[2]{\langle #1|#2\rangle}      %the inner product
\newcommand{\HC}{\mathcal{H}}

\renewcommand{\geq}{\geqslant}
\renewcommand{\leq}{\leqslant}

\newcommand{\bob}[0]{\textnormal{Bob}}

\newcommand{\ot}{\otimes}

\newcommand*{\id}{\openone}

\newcommand{\al}{\alpha }
\newcommand{\bt}{\beta }

\newcommand{\pguessmax}{p_{\textnormal{guess}}^{\textnormal{max}}}
\newcommand{\pguess}{p_{\textnormal{guess}}}
\newcommand{\Hmin}{H_{\textnormal{min}}}

\DeclarePairedDelimiter\abs{\lvert}{\rvert}%
\DeclarePairedDelimiter\norm{\lVert}{\rVert}%

\makeatletter
\let\oldabs\abs
\def\abs{\@ifstar{\oldabs}{\oldabs*}}
\let\oldnorm\norm
\def\norm{\@ifstar{\oldnorm}{\oldnorm*}}
\makeatother

\DeclareMathOperator{\Tr}{Tr}

\begin{document}
\title{Quantum preparation uncertainty and lack of information}
\author{Filip Rozp\k{e}dek}
\email{f.d.rozpedek@tudelft.nl}
\affiliation{QuTech, Delft University of Technology, Lorentzweg 1, 2628 CJ Delft, The Netherlands}

\author{J\k{e}drzej Kaniewski}
\affiliation{QuTech, Delft University of Technology, Lorentzweg 1, 2628 CJ Delft, The Netherlands}
\affiliation{Centre for Quantum Technologies, 3 Science Drive 2, 117543 Singapore}
\affiliation{QMATH, Department of Mathematical Sciences, University of Copenhagen, Universitetsparken 5, 2100 Copenhagen, Denmark}

\author{Patrick J. Coles}
\affiliation{Institute for Quantum Computing and Department of Physics and Astronomy, University of Waterloo, N2L3G1 Waterloo, Ontario, Canada}

\author{Stephanie Wehner}
\affiliation{QuTech, Delft University of Technology, Lorentzweg 1, 2628 CJ Delft, The Netherlands}
\begin{abstract}
The quantum uncertainty principle famously predicts that there exist measurements that are inherently incompatible, in the sense that their outcomes cannot be predicted simultaneously. In contrast, no such uncertainty exists in the classical domain, where all uncertainty results from ignorance about the exact state of the physical system. Here, we critically examine the concept of preparation uncertainty and ask whether similarly in the quantum regime, some of the uncertainty that we observe can actually also be understood as a lack of information (LOI), albeit a lack of \emph{quantum} information. We answer this question affirmatively by showing that for the well known measurements employed in BB84 quantum key distribution \cite{bb84}, the amount of uncertainty can indeed be related to the amount of available information about additional registers determining the choice of the measurement. We proceed to show that also for other measurements the amount of uncertainty is in part connected to a LOI. Finally, we discuss the conceptual implications of our observation to the security of cryptographic protocols that make use of BB84 states.
\end{abstract}
\maketitle
% !TEX root = main.tex

\section{Introduction}\label{sec:intro}
The uncertainty principle forms one of the cornerstones of quantum theory. As first observed by Heisenberg~\cite{heisenberg} and then rigorously proven by Kennard~\cite{kennard}, it is impossible to perfectly predict the measurement outcomes of both position and momentum observables. This notion was generalised by Robertson to an arbitrary pair of observables~\cite{Robertson} showing that uncertainty is an inherent feature of any non-commuting measurements in quantum mechanics. The described uncertainty is often referred to as preparation uncertainty, because it states that it is impossible to prepare a quantum state for which one could perfectly predict the measurement outcome of both observables.
\\ \indent
A modern way of capturing the notion of preparation uncertainty is by means of a \emph{guessing game}~\cite{Renner}. Such a game makes the concept of preparation uncertainty operational and is of great use in proving the security of quantum cryptographic protocols~\cite{URsurvey}.
Fig.~\ref{fig:game} summarises the game, which in its simplest form works as follows. Bob prepares system $B$ in an arbitrary state $\rho_B$ of his choosing and then passes it to Alice. Alice performs one of two incompatible measurements labeled by $r=0$ and $r=1$ according to a random coin flip contained in the register $R$ and obtains measurement outcome $X$. She then informs Bob which measurement she performed by sending him the register $R$. Bob wins the game if he correctly guesses Alice's measurement outcome $X$.
\\ \indent
To see why this captures the essence of the uncertainty principle, note that if the measurements are incompatible, then there exists no state $\rho_B$ that Bob can prepare that would allow him to guess the outcomes for both choices of measurements with certainty. Uncertainty can thus be quantified by a bound on the average probability that Bob correctly guesses $X$. That is, a relation of the form
\begin{align}
P_{\rm guess}(X|\bob) = p(r=0) P_{\rm guess}(X|\bob,r=0) \nonumber \\
\quad + p(r=1)P_{\rm guess}(X|\bob,r=1) \leq 2^{-\zeta}\, ,
\end{align}
for all $\rho_B$. Equivalently, we can relate the above defined guessing probability to the min-entropy $H_{\rm min}(X|\bob) = - \log P_{\rm guess}(X|\bob)$ (in this article all logarithms are base 2), so that we obtain an inequality:
\begin{align}
H_{\rm min}(X|\bob) \geq \zeta\, .
\end{align}
This expression forms an uncertainty relation as long as the RHS is non-trivial (i.e. $\zeta > 0$). Analogous relations exist for other entropies~\cite{URsurvey}, but here we focus on the min-entropy since it is the relevant measure for quantum cryptography and randomness generation, and it quantifies the winning probability for the aforementioned guessing game. 

In this work, we seek a deeper understanding of the uncertainty principle by considering a more general scenario than the typical guessing game and observing the conditions under which Bob's uncertainty vanishes. In particular, the generalisation we consider is to allow Bob to have additional information - possibly \textit{quantum} information - about Alice's measurement choice. This generalisation is closely related to recent proposals for quantum control experiments \cite{TernoQBS,TernoControl}. To elaborate, we note that Alice's random measurement choice in the guessing game can be implemented by preparing a qubit $R$ in the maximally mixed state $\rho_R = \mathbb{I}/2$ and then performing a unitary operation on $B$ conditioned on the state of $R$ (see Fig.~\ref{fig:circuit} below). In the generalised game that we consider, we allow $\rho_R$ to be a more general state, possibly with some coherence. As we discuss below, allowing for coherence in $\rho_R$ corresponds to giving Bob more information.

Our motivation for considering this scenario is to distinguish between uncertainty that is due to Bob's lack of information (LOI) versus uncertainty that is intrinsic or unavoidable. To help clarify these notions, we remark that a classical theory admits no intrinsic uncertainty. Classical here refers to commuting measurements that are jointly diagonal in one predefined basis. If Alice employed such measurements in the aforementioned guessing game, then the only way for her to prevent Bob from winning the game would be for her to add noise to her measurement outcomes, i.e., implement noisy measurements. Yet, we would classify Bob's uncertainty in this case as LOI uncertainty, as he simply lacks the information about the noise Alice adds. Hence, the arising uncertainty is clearly not an intrinsic feature of the measurements.
\\ \indent
Notice that preparing the register $R$ in the maximally mixed state $\rho_R = \mathbb{I}/2$ injects classical randomness into the protocol. It is unclear whether or not this randomness is ultimately responsible for the uncertainty principle, and this is a question we aim to answer. We emphasise that the scenario we consider differs from other variants of the uncertainty principle which derive bounds involving the purity or entropy of $\rho_B$ \cite{Renner, sanchez, berta2014entanglement, hall, christandl, renes, dupuis, Berta, liu, ColesUncertaintyRel, frank, furrer, ColesTripartite, luo, korzekwa}. 
\\ \indent
Interestingly we find that in the special case where Bob's system is a qubit ($d=2$), there is no intrinsic uncertainty but all the uncertainty is due to LOI. That is, if Bob has complete knowledge about the preparation of $R$ (i.e., $R$ is in a pure state), then his uncertainty vanishes. In contrast, for all dimensions $d>2$, we find that there is always some intrinsic uncertainty. That is, even with the full knowledge about the preparation of $R$, Bob cannot win the guessing game with unit probability. Before we discuss these results in detail, let us outline the physical setup.
\begin{figure}[t!]
  \centering
    \includegraphics[trim = 0 190 0 250, clip, width=0.5\textwidth]{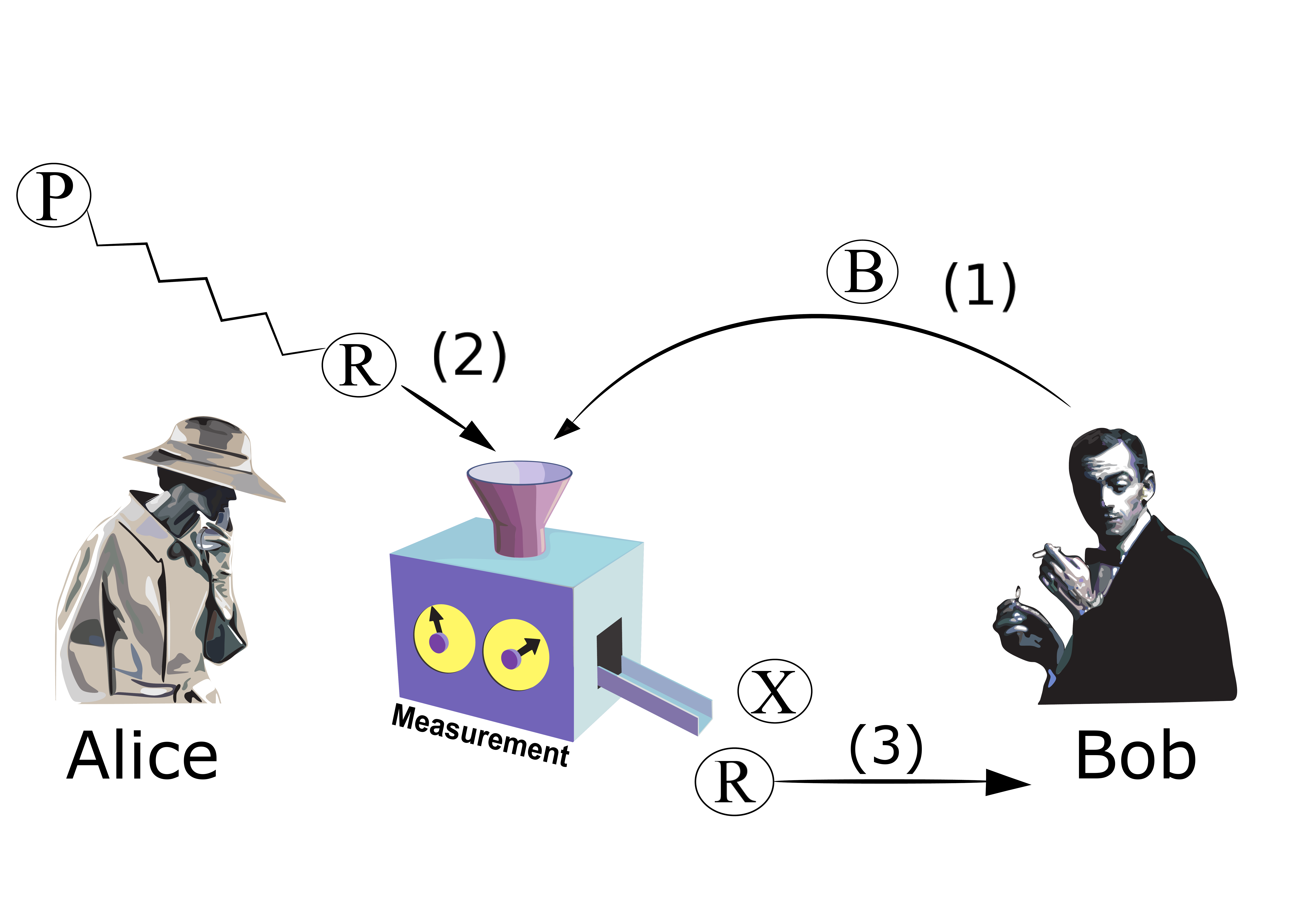}
    \caption{Uncertainty guessing game. 
The game runs as follows: (1) First, Bob prepares system $B$ in a state $\rho_B$ and sends it to Alice.  We show in Appendix~\ref{sec:appendixA} that Bob's best strategy is to prepare a pure state $\rho_B = \dya{\phi}_B$. (2) Second, Alice measures $B$ in a basis determined by the state of register $R$. (3) Finally, Alice obtains the classical outcome $X$ and sends $R$ to Bob. Bob can then measure $R$ in order to help him guess $X$. Note that $R$ may be initially prepared in a mixed state $\rho_R$, and Bob does not have access to the purifying system of $\rho_R$, denoted as $P$ in the figure. Hence, $P$ embodies Bob's lack of information in this game.}
    \label{fig:game}
\end{figure}

\section{Physical setup}
\subsection{Degrees of ignorance}
In this section we describe the generalised guessing game shown in Fig.~\ref{fig:game}. Here, Alice prepares a register system $R$ in some state $\rho_R$. Meanwhile Bob prepares system $B$ in state $\rho_B$ and sends it to Alice. Alice measures $B$ in a basis determined by the state of $R$. Then she passes $R$ to Bob, and he tries to guess her measurement outcome, possibly using the information stored in $R$. We are interested in understanding how much of Bob's uncertainty (i.e., his inability to win this game) is due to LOI and how much corresponds to intrinsic (or unavoidable) uncertainty. 

To better understand this, let us examine what Bob does and does not have access to in Fig.~\ref{fig:game}. Since $\rho_R$ is generally a mixed state, it can be purified by considering an additional system, $P$. Even though Bob is given access to $R$, we emphasise that he does not have access to $P$ in our guessing game. Hence, we can think of $P$ as representing Bob's LOI.

For example, consider the case when $\rho_R = \id/2$ is maximally mixed, which corresponds to the case where the measurement choice is a classical coin flip (i.e., the typical uncertainty game considered in the literature \cite{Renner}). The purification is a maximally entangled state such as
\begin{align}
\ket{\xi_{RP}} = \frac{1}{\sqrt{2}}\left(\ket{0}_R \ket{0}_P + \ket{1}_R \ket{1}_P\right)\, .
\end{align}
At the other extreme is the case where $\rho_R$ is pure, i.e.,
\begin{align}
\ket{\xi_{RP}}=\ket{\xi_R} \otimes \ket{\xi_P}
\end{align}
is a product state. We will take $\ket{\xi_R} = \frac{1}{\sqrt{2}}\left(\ket{0} + \ket{1}\right)$, i.e., we choose an equal superposition in correspondence with the idea that both measurements were previously chosen with equal probability. Intuitively, when the initial state is maximally entangled, then Bob will later suffer from a maximum LOI about $P$. However, in the case where the two systems are uncorrelated, Bob does not need $P$ at all. In other words, there is no LOI on his part, because $R$ is pure. 
\\ \indent
There are many ways to interpolate between these two extremes in terms of a measure of correlation between $R$ and $P$. Here, we choose one that is intuitive when we think about ``how much'' of $P$ Bob is actually lacking. Concretely, we imagine that apart from the classical coin $C$ (which is a part of $R$), $R$ and $P$ are actually comprised of many environmental subsystems $E_1,\ldots,E_n$, and we quantify Bob's LOI by the number of the environment systems that are part of $P$ instead of part of $R$.  
Specifically, we take
\begin{align}
\ket{\xi_{RP}} = \frac{1}{\sqrt{2}}\left(
\ket{0}_C \otimes \bigotimes_{i=1}^n \ket{\al}_{E_i}+
\ket{1}_C \otimes \bigotimes_{i=1}^n \ket{\bt}_{E_i}\right)\, ,
\label{eq:subsystems} 
\end{align}
where $RP = CE_1\ldots E_n$. The environments $E_{j}$'s are two-dimensional registers and $\abs{\braket{\al|\bt}} = 1-\epsilon$, with $\epsilon >0$ and $\epsilon \ll 1$ so that each individual $E_{j}$ holds very little information about the state of the coin $C$. However, we see that $\ip{\al}{\bt}^n \to 0$ as $n \rightarrow \infty$.
We thus see that for $n \rightarrow \infty$ and $R=C$, $P=E_1\ldots E_n$, we approach the extreme case of $R$ being essentially classical, and $\ket{\xi_{RP}}$ being
maximally entangled. This idea of approximating the notion of a classical register by ``copying'' information into a large number of environmental systems $E_j$ is due to Zurek~\cite{Zurek}. 

We can now interpolate between the two extremes by letting $R = CE_1\ldots E_j$ and $P = E_{j+1}\ldots E_n$. 
We have that
\begin{align}
\rho_R  = \frac{1}{2}(\dya{0}+\dya{1}+\gamma^* \dyad{0}{1} + \gamma \dyad{1}{0})\, ,
\label{eq:rhoR}
\end{align}
where
\begin{align}
\ket{0}_R	&:= \ket{0}_C\ot \bigotimes_{i=1}^j \ket{\al}_{E_i}\, ,\\
\ket{1}_R	&:= \ket{1}_C\ot \bigotimes_{i=1}^j \ket{\beta}_{E_i}\, ,\\
\gamma 		&= \ip{\al}{\bt}^{n-j}\,.
\label{eq:ngammameaning}
\end{align}
We see that $\abs{\gamma}$ increases monotonically with $j$, the number of environmental subsystems contained in $R$, and hence the number of subsystems to which Bob is given access later on. The extreme cases $\gamma =0$ and $\gamma = 1$ correspond respectively to $j=0$ and $j =n$ (again note that the number of environment subsystems is very large so that we always consider the limit $n \to \infty$).
In Appendix~\ref{sec:appendixA} we show that for the uncertainty game it is only the modulus of $\gamma$ that matters. Therefore, we will only consider the case of real and positive $\gamma$, i.e.~$\gamma \in [0,1]$.

\subsection{Uncertainty game}

Let us now revisit our uncertainty guessing game (see Fig.~\ref{fig:game} and Fig.~\ref{fig:circuit}) with a more detailed description. First, Bob prepares system $B$ in a state $\rho_B$ and sends it to Alice. Second, Alice measures $B$ and obtains the classical outcome $X$, with the measurement basis determined by the state of register $R$ given by: 
\begin{align}
\rho_R  = \frac{1}{2}(\dya{0}+\dya{1}+\gamma \dyad{0}{1} + \gamma \dyad{1}{0})\, .
\label{eq:rhoR2}
\end{align} 
Specifically, as depicted in Fig.~\ref{fig:circuit}, states $\ket{0}$ and $\ket{1}$ on $R$ are, respectively, associated with measuring in the standard basis and Fourier basis on $B$ (we have chosen maximally incompatible bases to maximise the ``inherent'' uncertainty). Next, Alice sends Bob the register $R$. Finally Bob measures $R$ to help him produce a guess for $X$. This defines a two-parameter family of uncertainty games which depend on: $d \in \{2,3,\ldots\}$, the number of possible outcomes (which fixes the dimension of the quantum state $\rho_B$ supplied by Bob and the dimension of the Fourier transform in Fig.~\ref{fig:circuit}) and $\gamma \in [0,1]$, describing the amount of information about $R$ that is held in $P$, or equivalently the amount of coherence in $R$.

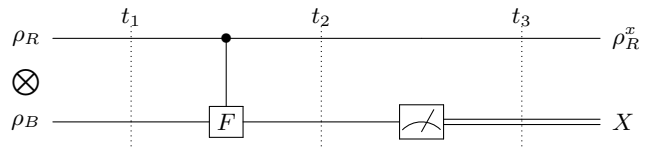
\begin{figure}[t]
\vspace{4mm}
\centerline{
\Qcircuit @C=3.2em @R=1.3em {
\lstick{\rho_R} & \ustick{t_1} \qw \ar@{.}[]+<0em,0.5em>;[d]+<0.em,-3em>  & \ctrl{2} &\ustick{t_2} \qw \ar@{.}[]+<0em,0.5em>;[d]+<0.em,-3em> &\qw &\ustick{t_3} \qw \ar@{.}[]+<0em,0.5em>;[d]+<0.em,-3em>  & \rstick{\rho_R^x} \qw  \\
\lstick{\raisebox{-1.5em}{$\bigotimes$}} & & & & & & \\
\lstick{\rho_B} & \qw & \gate{F} & \qw & \meter & \cw &\rstick{X} \cw
}}
\vspace{1mm}
\caption{Quantum circuit of the uncertainty game.
At time $t_1$, Alice's register $R$ and Bob's system $B$ are uncorrelated. 
We will assume that Alice measures in the standard basis and one additional basis depending on the state of register $R$. 
To allow for maximum intrinsic uncertainty, we take the other basis to be maximally incompatible. Here, we choose it to be the Fourier basis. Hence the two measurements correspond to measuring in two mutually unbiased bases. If $B$ is a qubit,
then this means that Alice measures in the standard and Hadamard basis, which are the two bases used in BB84 quantum key distribution. This basis choice is performed by Alice applying a controlled unitary between the two registers, leading to a correlated state at time $t_2$. Alice then measures $B$ to obtain the measurement outcome $X$. If the register $R$ is classical, then the two operations together correspond to performing a random measurement. If the register $R$ contains some non-zero coherence, then those operations describe a procedure which could be understood as a ``measurement in a superposition of two bases''. After time $t_3$, Alice sends $R$ to Bob. At this stage, $\rho_{RX} = \sum_x p_x \rho_R^x \otimes \ketbra{x}{x}_X$ is a qc-state. Bob can then make a measurement in order to distinguish the states $\rho_R^x$, i.e., to help him guess $X$. 
Note that Bob knows which states $\rho_R^x$ he wants to distinguish since he knows the form of the initial
state $\ket{\xi_{RP}}$ and the measurements Alice can perform.
}
\label{fig:circuit}
\end{figure}

\section{Methods}\label{sec:methods}
Here we provide a high level overview of the methods used to obtain the results presented in the next section. For complete analysis we refer the reader to the appendices.

After Alice has performed her measurement, at time $t_3$ in Fig.~\ref{fig:circuit} the resulting qc-state between the register $R$ and the outcome register $X$ is:
\begin{equation}
\rho_{RX} (\gamma,d, \rho_B) = \sum_{x} \tilde{\rho}^x_R (\gamma, d, \rho_B) \otimes\dya{x}_{X} \, ,
\end{equation}
where $\tilde{\rho}^x_R (\gamma, d, \rho_B) = p_x(d, \rho_B) \rho^x_R (\gamma, d, \rho_B)$ is the subnormalised post-measurement state of the register $R$ corresponding to the outcome $X=x$. In terms of Bob's input state $\rho_B$, this state has the form:
\begin{equation}
\begin{aligned}
\tilde{\rho}^x_R (\gamma, d, \rho_B) 	&= 
\frac{1}{2}   \left( \begin{array}{cc}
\bra{x}\rho_B\ket{x} & \gamma \bra{x}\rho_B F^\dag\ket{x} \\
\gamma \bra{x} F \rho_B\ket{x} & \bra{x}F\rho_BF^\dag \ket{x}
\end{array} \right)\, ,
\label{eq:rho_xmixed}
\end{aligned}
\end{equation}
as derived in Appendix~\ref{sec:appendixA}. Since Bob later gains access to register $R$, we see that in order to guess the resulting outcome $X=x$, Bob should try to determine which quantum state $\rho^x_R (\gamma, d, \rho_B)$ he has received. Hence, his guessing problem becomes equivalent to the problem of distinguishing quantum states $\{\rho^x_R (\gamma, d, \rho_B)\}$ occurring with probabilities $\{p_x(d, \rho_B)\}$.

The probability of Bob correctly discriminating those states with the optimal strategy, i.e., with the optimal measurement on $R$ (described by POVM elements $\{M_x\}$), is given by~\cite{Watrous}:
\begin{equation}
\pguess(\gamma, d, \rho_B) = \max_{\{M_x\}}\sum_{x=0}^{d-1}p_x (d, \rho_B) \Tr[M_x \rho^x_R (\gamma, d, \rho_B)] \, .
\label{eq:pguess1}
\end{equation}
In Appendix~\ref{sec:appendixA} we show that to achieve $\pguessmax(\gamma, d)$, the guessing probability optimised over input states $\rho_B$, it is sufficient to consider only pure input states $\rho_B = \dya{\phi}_B$. Hence, the maximum value of $\pguess(\gamma, d, \rho_B)$ for a given $\gamma$ and $d$ is the result of optimising the guessing probability over all input states $\ket{\phi}_B$ of Bob (for convenience we will often omit the subscript ``$B$'' from $\ket{\phi}_B$). That is,
\begin{equation}
\pguessmax(\gamma,d) = \max_{\ket{\phi}} \pguess(\gamma,d,\ket{\phi})\,.
\label{eq:pguessmax}
\end{equation}
Solving this optimisation problem is not an easy task. Note that the function which we want to optimise over all the POVM elements $\{M_x\}$ in Eq.~\eqref{eq:pguess1} is linear in those operators. Hence, for a specific input state $\ket{\phi}_B$ the optimisation can be performed using techniques of semi-definite programming. However, the above optimisation problem in Eq.~\eqref{eq:pguessmax} involves optimisation both over POVM elements and input states $\ket{\phi}_B$. Clearly, $\tilde{\rho}^x_R (\gamma, d, \ket{\phi}_B)$ is quadratic in $\ket{\phi}_B$. Note that this problem can be made linear in the input state by again considering optimisation over all mixed states $\rho_B$, i.e. our problem is then linear in $\rho_B$. However, the full problem of optimising over both $\{M_x\}$ and $\rho_B$:
\begin{equation}
\pguessmax(\gamma, d) = \max_{\rho_B} \max_{\{M_x\}}\sum_{x=0}^{d-1}p_x (d, \rho_B) \Tr[M_x \rho^x_R (\gamma, d, \rho_B)]
\end{equation}
turns out not to be jointly concave in both of those variables and so cannot be solved using techniques of convex optimisation.
\subsection{Two-dimensional game}
Nevertheless, we can solve this problem analytically for $d=2$. For this case, we derived our result (stated below in Theorem~\ref{pguessmaxd2}) by noting that the problem of optimising over the POVM elements in Eq.~\eqref{eq:pguess1} (for fixed states $\{\rho^x_R\}$ occuring with fixed probabilities $\{p_x\}$) has been solved analytically by Helstrom~\cite{Helstrom}:
\begin{equation}
\pguess(\gamma, d=2, \rho_B) = \frac{1}{2} \Big(1 + \norm{ \tilde{\rho}^0_R (\gamma, \rho_B) - \tilde{\rho}^1_R(\gamma, \rho_B)   }_1\Big)\, ,
\label{eq:pguessmaxqubithelstrom}
\end{equation}
where $\norm{\cdot}_1$ denotes the trace norm and we have omitted the $d=2$ argument in $\tilde{\rho}^0_R$ and $\tilde{\rho}^1_R$. In this way we obtain an expression for $\pguess(\gamma, d=2, \rho_B)$ which we then analytically optimise over the input states $\rho_B$ for every value of $\gamma \in [0,1]$ to obtain $\pguessmax(\gamma, d=2)$ (see Appendix~\ref{sec:appendixB}). For completeness, we still optimise over all qubit states $\rho_B$, not only the pure ones. This allows us to find all the qubit input states that achieve $\pguessmax(\gamma, d=2)$.
\subsection{Higher-dimensional games}\label{sec:highdimgames}
For $d>2$ we cannot calculate $\pguessmax(\gamma, d>2)$ analytically, since there exists no known analytical expression for the probability of correctly distinguishing more than two quantum states. However, we can find $\pguess(\gamma, d, \ket{\phi})$ for an arbitrary state $\ket{\phi}$ using techniques from semi-definite programming. We obtain numerical lower bounds for $\pguessmax(\gamma, d)$, shown in Fig.~\ref{fig:p_guess_max_vs_gamma_different_d1}, by solving a semi-definite programme for $\pguess(\gamma, d, \ket{\phi})$ and numerically searching for local maxima of $\pguess(\gamma, d, \ket{\phi})$ with respect to the input state $\ket{\phi}$ using the Nelder-Mead algorithm. We repeat the search multiple times with a randomly generated initial state in each run, that is drawn uniformly from unit vectors on $\mathbb{C}^d$.

\section{Results}

In Section~\ref{sec:intro} we discussed that classical uncertainty arises solely from LOI. Here we show that even in the quantum case, uncertainty can in part be understood as a LOI that Bob has - namely a lack of quantum information about the register $P$. For the case of $d=2$ and BB84 measurements as they are used in quantum key distribution (QKD), this effect is indeed dramatic. We find (see Theorem~\ref{pguessmaxd2} below) that there is no more uncertainty at all in the case where $R$ is pure and $P$ is uncorrelated, meaning that Bob does not suffer from any LOI.

First, we consider the typical uncertainty game where $R$ is a classical coin, i.e., $R$ and $P$ are maximally entangled ($\gamma = 0$). In this case the maximum value of the guessing probability (for completeness derived in Appendix~\ref{sec:appendixC}) is given by:
\begin{equation}
\pguessmax(\gamma=0, d) = \frac{1}{2}\left(1 + \frac{1}{\sqrt{d}}\right)\, .
\label{eq:methodspguessmaxgamma0}
\end{equation}
The states $\rho_B$ that achieve the guessing probability of Eq.~\eqref{eq:methodspguessmaxgamma0} are the pure states
\begin{equation}
\ket{\phi_{jl}} := c(\ket{j} + \omega^{jl}F^\dag\ket{l})\, ,
\label{eq:optimalstatesphijl}
\end{equation}
where $c = \sqrt{\sqrt{d}/(2 \sqrt{d}+ 2 )}$ is the normalisation constant, $F$ denotes a quantum Fourier transform defined in Appendix~\ref{sec:appendixA}, $\omega$ is the $d$-th root of unity and $j$ and $l$ are integer indices that lie in the range $\{0, 1, \ldots, d-1\}$ so that the pure states $\ket{j}$ and $\ket{l}$ denote the corresponding eigenstates of the standard basis. The states defined in Eq.~\eqref{eq:optimalstatesphijl} are the states where the dominant classical outcome for the measurement is $j$ in the standard basis and $l$ in the Fourier basis.

Now we consider the more general case where $R$ may have some coherence. For $d = 2$ we have found the analytical solution for all $\gamma \in [0, 1]$. In this case the guessing probability is equal to the probability of successfully distinguishing the two possible post-measurement states of the basis register, namely $\rho^0_R$ and $\rho^1_R$ corresponding to outcomes 0 and 1 respectively (see Fig.~\ref{fig:circuit}).
\begin{theorem}\thlabel{pguessmaxd2}
The maximum guessing probability for a two-dimensional game ($d=2$), optimised over all input states $\rho_B$ is given by:
\begin{equation}
\pguessmax(\gamma, d=2) = \frac{1}{2} \left(1 + \frac{\sqrt{2+2\gamma^2}}{2}\right)\, .
\label{eq:pguessmaxqubit1}
\end{equation}
In particular, for $\gamma = 1$ one achieves perfect guessing, that is $\pguessmax(\gamma=1, d=2) = 1$.
\end{theorem}
It is also possible to express this guessing probability in terms of the purity of the basis register:
\begin{equation}
\pguessmax(\gamma, d=2) = \frac{1}{2} \left(1 + \sqrt{\Tr[\rho_R^2]}\right)\, .
\end{equation}
For all $\gamma \in [0,1]$, this guessing probability can be achieved by one of two orthogonal input states of Bob, $\ket{\phi_{01}} = c(\ket{0} + \ket{-})$ and $\ket{\phi_{10}} = c(\ket{1} + \ket{+})$, which are mapped by the Hadamard transformation onto each other. (For $\gamma=0$ this guessing probability can of course also be achieved by $\ket{\phi_{00}}$ and $\ket{\phi_{11}}$, as then Eq.~\eqref{eq:pguessmaxqubit1} reduces to Eq.~\eqref{eq:methodspguessmaxgamma0}. For $\gamma=1$ the optimal input states form a continuous one-parameter family, see Appendix~\ref{sec:appendixB}.) 

From Eq.~\eqref{eq:pguessmaxqubit1} we see that Bob can achieve perfect guessing probability for the case when $R$ is uncorrelated from $P$ (and so $P$ holds no information about $R$ and there is no LOI about the measurement process on Bob's side). This is connected to the fact, that for ${\gamma = 1}$ and a suitable choice of input state $\rho_B$, the joint state  $\rho_{RB}$ becomes maximally entangled at time $t_2$ just before Alice's measurement in Fig.~\ref{fig:circuit} (see Appendix~\ref{sec:appendixD} below for discussion of this connection). The above results for $d=2$ are derived in Appendix~\ref{sec:appendixB}.

Now it is interesting to ask what happens to the measurement uncertainty in the game with more than two measurement outcomes in higher dimension. 
It is intuitive that the dramatic effect we see for $d=2$ should be less prominent here. After all, Bob is trying to guess measurement outcomes that
can take on $d$ values, while $R$ and $P$ each remain two-dimensional and can hence only contain limited information about the outcomes. 
We first make this intuition precise in the following theorem.
\begin{theorem}\thlabel{perfectguessinghighdim}
For $d$-dimensional games with any $d>2$ it is not possible to achieve perfect guessing, i.e.,
\begin{equation}
\pguessmax(\gamma, d>2) < 1 \, , \quad\quad \forall \hspace{2pt}\gamma \in [0,1]\,.
\label{eq:pguessmaxqubit111}
\end{equation}
\end{theorem}

Crucially, however, coherence in register $R$ \emph{always} facilitates guessing.
\begin{theorem}\thlabel{coherencevsnocoherence}
For $d$-dimensional games with $d$ being arbitrary, the maximum guessing probability when $R$ has any non-zero amount of coherence is always strictly greater than the case of maximally mixed $R$. That is,~for all $\gamma' > 0$
\begin{equation}
\pguessmax(\gamma = \gamma', d) > \pguessmax(\gamma = 0, d) \, , \quad\quad \forall \hspace{2pt}d \ge 2\,.
\label{eq:pguessmaxqubit112}
\end{equation}
\end{theorem}
Moreover, we show that for a subclass of the input states that are optimal for $\gamma=0$, the guessing probability monotonically increases with $\gamma$.  Specific values of $\pguessmax(\gamma, d)$ are lower bounded numerically. Those results are depicted in Fig.~\ref{fig:p_guess_max_vs_gamma_different_d1}.
\begin{figure}[t]
\centering
    \includegraphics[width=0.48\textwidth]{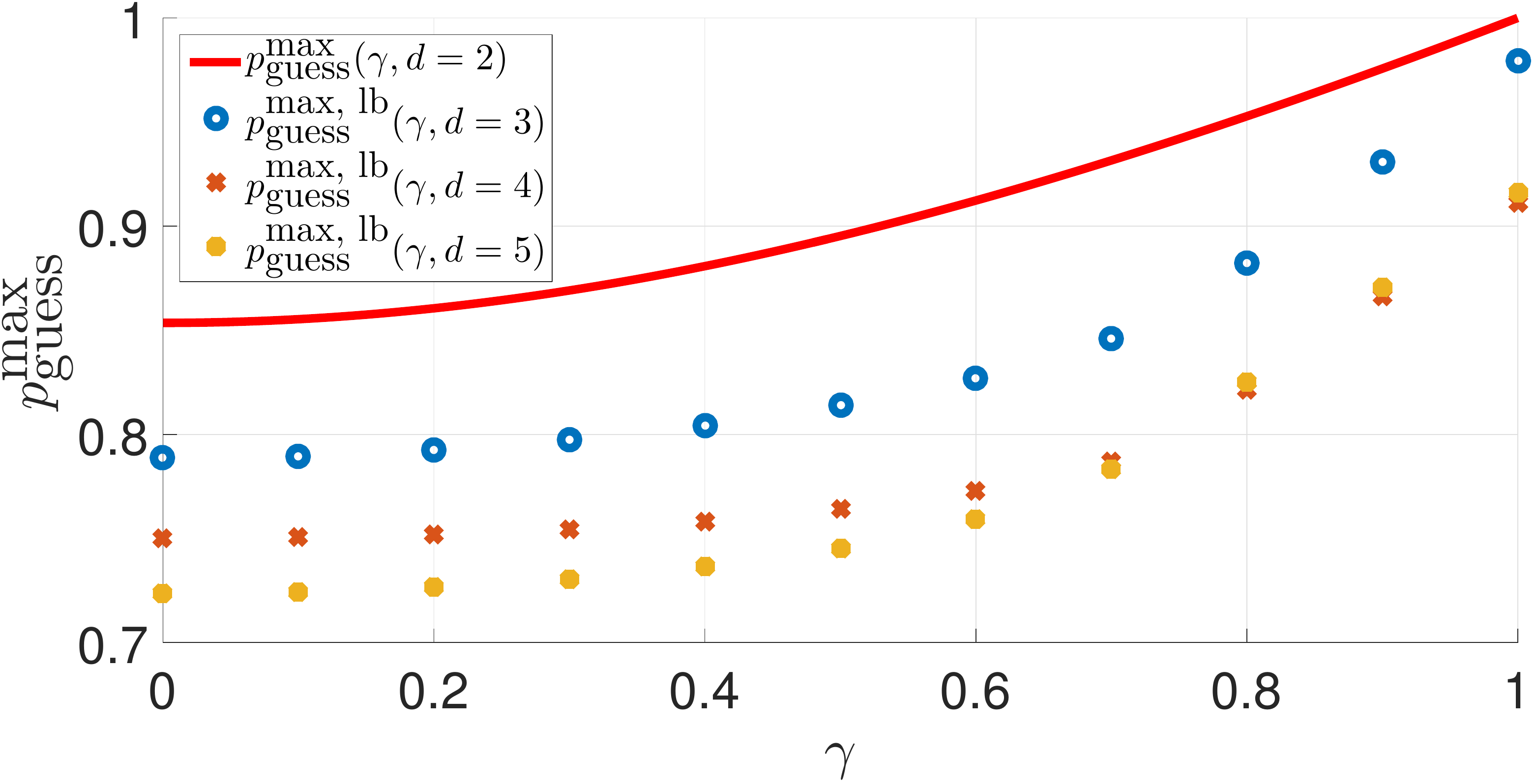}
    \caption{The optimal guessing probabilities $\pguessmax(\gamma,d)$ as a function of $\gamma$ for different $d$. The solid line corresponds to the analytical solution $\pguessmax(\gamma, d=2)$ for a two-dimensional game. The remaining data corresponds to the numerical lower bounds $p_{\text{guess}}^{\text{max, lb}}(\gamma, d)$ for $d= 3,4,5$. For $\gamma = 0$ the numerical values coincide with the analytical solution $\pguessmax(\gamma=0, d) = \frac{1}{2}\left(1 + \frac{1}{\sqrt{d}}\right)$. The crossing of the dotted lines corresponding to $d=4$ and $d=5$ is discussed in Section~\ref{sec:discussion}.}
		\label{fig:p_guess_max_vs_gamma_different_d1}
\end{figure}

\section{Discussion}\label{sec:discussion}

We have shown that quantum preparation uncertainty is not always inherent to the measurement process but on the contrary it depends on the amount of information that one has about this process. In particular, for $d=2$, if Bob has all the information about the measurement process, then he can perfectly predict the measurement outcome. In the cryptographic protocols that use BB84 states, $\rho_R$ is a maximally mixed state. Hence, from the perspective of cryptographic security, this shows that it is important for the purification of $\rho_R$ to remain inaccessible to the adversary. In particular, the more of the purification $P$ becomes incorporated into $R$, the larger the guessing probability becomes and so the more the security of our cryptographic protocols becomes compromised. Passive encoding schemes \cite{curty2010}, which generate the QKD signal states by performing a measurement on a quantum register (analogous to our $R$), would especially need to consider this issue.

On the other hand, we found that there is always some unavoidable uncertainty for guessing games in higher dimensions, $d>2$. This result is somewhat intuitive when one considers that our guessing game allows for two measurements, and hence system $R$ is only two-dimensional. The intuition behind this unavoidable uncertainty is that the state $\rho_R$, in which the information about the measurement outcome becomes encoded, is always a qubit, while the number of outcomes is $d$. Hence, even if Bob inputs a state that results in entanglement between the two systems, this entanglement lives in a two-dimensional subspace of the $d$-dimensional space $\HC_B$. Therefore, the joint state cannot be maximally entangled and since the Fourier transformation applied to elements of the standard basis generates a basis that is unbiased to it, the correlations before the measurement of Alice do not align with the standard basis in which the measurement is performed. This fact can also be seen by noting that perfect guessing could only occur if only two of the resulting outcomes had non-zero probability and if those outcomes produced orthogonal post-measurement states of the register $R$. It turns out that all those conditions cannot be met simultaneously.

The crossing of the dotted lines corresponding to $d=4$ and $d=5$ in Fig.~\ref{fig:p_guess_max_vs_gamma_different_d1} is an interesting phenomenon. We have investigated it extensively using multiple methods and numerical solvers on which we now elaborate. As mentioned in Section~\ref{sec:methods} the problem of optimisation over both input states and measurements is in general very hard because the optimisation problem that we face is not convex. That is we can have no guarantee that the solution that we find is the global maximum. Therefore the numerical results are just the lower bounds on the $\pguessmax$, as they represent achievable values of $\pguessmax$ that have been found. Nevertheless we have used multiple methods to look for these optimal bounds. Apart from the method described in Section~\ref{sec:highdimgames} (where part of the data was checked by rerunning the programme with multiple numerical solvers), we have tried imposing a net over the statespace and solving the semi-definite programme over the measurements for each of those states. Then the procedure was repeated with a denser net in the region where the highest guessing probability has been found. This step of ``zooming-in'' has then been repeated multiple times. Finally we have also used the ``Penlab'' solver, which can also provide achievability bounds for non-linear problems. Application of those other methods however resulted in much worse bounds and so they shed no light on the nature of the crossing in Fig.~\ref{fig:p_guess_max_vs_gamma_different_d1}.

Nevertheless, despite the fact that we only find achievable bounds, we believe that the crossing seen in Fig.~\ref{fig:p_guess_max_vs_gamma_different_d1} could in principle arise even for the exact solution. We note that while asymptotically we expect $\pguessmax(\gamma, d)$ to tend to 0.5 as $d$ tends to infinity, it is possible for $\pguessmax(\gamma, d)$ to be larger for $d=5$ than for $d=4$ above some threshold $\gamma = \gamma_0$. As we mentioned earlier, the optimal guessing probability depends on the optimal correlations between two-dimensional register $R$ and $d$-dimensional register $B$. The resulting state is asymmetric and so it is possible that certain favourable correlations are possible for $d=5$, while not possible for $d=4$. The complexity of the problem can be seen by looking at the Schmidt coefficients of the joint state of registers $R$ and $B$ at time $t_2$ in Fig.~\ref{fig:circuit}. For $d=2$ and $\gamma = 1$ the optimal input states are precisely the ones that lead to a maximally entangled state between those two registers at time $t_2$. One might intuitively guess that also for $d>2$ forming maximally entangled states within the two-dimensional subspace of $B$ will lead to the optimal guessing probability for $\gamma=1$. This turns out not be sufficient: we checked specific states that lead to maximal entanglement in dimensions $d = 3, 4, 5$ and their performance is suboptimal. At the same time, all the optimal input states found numerically that achieve $p_{\text{guess}}^{\text{max, lb}}(\gamma=1, d)$ for $d = 3, 4, 5$ lead to unbalanced Schmidt coefficients. While we have found multiple states that achieve $p_{\text{guess}}^{\text{max, lb}}(\gamma=1, d)$ for each of $d = 3, 4, 5$, all of them lead to exactly the same Schmidt coefficients of the joint state, which we list in Table~\ref{table:Schmidt}. This fact, together with the irregularity of our numerical curves, reveals the complexity of the geometry of this problem.

\begin{center}
\begin{table}
\begin{tabularx}{0.3\textwidth}{c *{1}{|Y} c *{2}{Y}}
& \multicolumn{2}{c}{Schmidt coefficients} \\
\hline
$d = 3$ & $0.8122$ & $0.5834$ \\ 
$d = 4$ & $0.8314$ & $0.5556$ \\
$d = 5$ & $0.7415$ & $0.6709$ \\
\end{tabularx}
\caption{Schmidt coefficients of the joint state on $RB$ at time $t_2$ for the input states that achieve $p_{\text{guess}}^{\text{max, lb}}(\gamma=1, d)$.}
\label{table:Schmidt}
\end{table}
\end{center}

In future work, it would be very natural to consider games with more than two measurements. It would be interesting to investigate whether a higher dimensional register $R$ could then encode more information about the measurement outcome. Specifically, for the scenario with $d$ mutually unbiased measurements (if they exist) and $d$ possible outcomes, it is reasonable to ask whether one can again achieve perfect guessing (e.g., due to the possibility of creating maximal entanglement between $R$ and $B$).

Another natural extension of our game would be to provide Bob with access to a quantum memory \cite{Renner}. In such a scenario an interesting task would be to investigate the effect of the trade-off between Bob's amount of accessible information about the measurement process and the quality of entanglement between $B$ and Bob's quantum memory. 

Finally, we would like to emphasise that while the described guessing game seems to be only an abstract tool that we use to investigate the connection between quantum preparation uncertainty and lack of information, the game described in Fig.~\ref{fig:game} could in fact be implemented experimentally, e.g., using a Mach-Zehnder interferometer for single photons. For simplicity consider the case $d=2$, although the following discussion can be extended to $d>2$ by considering an interferometer with more than two paths. Suppose that system $R$ is the photon's polarisation, while $B$ is the photon's spatial degree of freedom (the path that it takes in the interferometer). Allowing Bob to have access to the first variable beam splitter of the interferometer allows him to prepare an arbitrary pure qubit state $\rho_B$ inside the interferometer (Bob is allowed to freely choose the reflectance and the relative phase of the beam splitter). The controlled Fourier transform in Fig.~\ref{fig:circuit} is implemented by making the second beam splitter of the interferometer a so-called quantum balanced beam splitter~\cite{TernoQBS}. That is, the photon's polarisation controls whether or not the balanced (50/50) beam splitter appears in the photon's path. Hence, this beam splitter can be effectively in a superposition of being absent and present, if one chooses the polarisation to be in a superposition. This would be a so-called quantum control experiment~\cite{TernoControl}. Let us note that such a quantum beam splitter has been implemented experimentally~\cite{QBSExp1,QBSExp2, QBSExp3}. The winning condition of the game for Bob is correctly guessing which one of the two photon detectors clicked, after being able to measure the polarisation state of the photon behind the quantum beam splitter.

\acknowledgements
We would like to greatly thank Jonas Helsen and Le Phuc Thinh for help with convex optimisation techniques. We are thankful to Pim Veldhuisen and Wouter Uijens for help with semi-definite programme implementation, to Thomas Schiet and Doru Sticlet for useful discussions on numerical sampling techniques and to David Elkouss, Kenneth Goodenough, Corsin Pfister and Mark Steudtner for valuable comments on the manuscript. We are also very grateful to Dmytro Vasylyev for preparing Fig.~\ref{fig:game}. FR and SW are funded by STW, NWO VIDI and an ERC Starting Grant. JK is supported by Ministry of Education, Singapore and the European Research Council (ERC Grant Agreement 337603). PJC is funded by Industry Canada, Sandia National Laboratories, Office of Naval Research, NSERC Discovery Grant, and Ontario Research Fund.

\bibliographystyle{arxiv2}
\bibliography{bibliography}{}

\begin{thebibliography}{10}

\bibitem{bb84}
C.~H. Bennett and G.~Brassard.
\newblock Quantum cryptography: Public key distribution and coin tossing.
\newblock {\em International Conference on Computer System and Signal
  Processing, IEEE}, 1984.

\bibitem{Renner}
M.~Berta, M.~Christandl, R.~Colbeck, J.~M. Renes, and R.~Renner.
\newblock The uncertainty principle in the presence of quantum memory.
\newblock {\em Nat Phys}, 6(9), 09 2010.
\newblock
  \texttt{\href{http://dx.doi.org/10.1038/nphys1734}{DOI:\,10.1038/nphys1734}}.

\bibitem{berta2014entanglement}
M.~Berta, P.~J. Coles, and S.~Wehner.
\newblock Entanglement-assisted guessing of complementary measurement outcomes.
\newblock {\em Phys. Rev. A}, 90:062127, Dec 2014.
\newblock
  \texttt{\href{http://dx.doi.org/10.1103/PhysRevA.90.062127}{DOI:\,10.1103/PhysRevA.90.062127}}.

\bibitem{Berta}
M.~Berta, O.~Fawzi, and S.~Wehner.
\newblock Quantum to classical randomness extractors.
\newblock {\em IEEE Transactions on Information Theory}, 60(2):1168--1192, Feb
  2014.
\newblock
  \texttt{\href{http://dx.doi.org/10.1109/TIT.2013.2291780}{DOI:\,10.1109/TIT.2013.2291780}}.

\bibitem{TernoControl}
L.~C. C{\'e}leri, R.~M. Gomes, R.~Ionicioiu, T.~Jennewein, R.~B. Mann, and
  D.~R. Terno.
\newblock Quantum control in foundational experiments.
\newblock {\em Foundations of Physics}, 44(5):576, 2014.
\newblock
  \texttt{\href{http://dx.doi.org/10.1007/s10701-014-9792-2}{DOI:\,10.1007/s10701-014-9792-2}}.

\bibitem{christandl}
M.~Christandl and A.~Winter.
\newblock Uncertainty, monogamy, and locking of quantum correlations.
\newblock {\em IEEE Transactions on Information Theory}, 51(9):3159--3165, Sept
  2005.
\newblock
  \texttt{\href{http://dx.doi.org/10.1109/TIT.2005.853338}{DOI:\,10.1109/TIT.2005.853338}}.

\bibitem{URsurvey}
P.~J. Coles, M.~Berta, M.~Tomamichel, and S.~Wehner.
\newblock Entropic uncertainty relations and their applications.
\newblock {\em Rev. Mod. Phys.}, 89:015002, Feb 2017.
\newblock
  \texttt{\href{http://dx.doi.org/10.1103/RevModPhys.89.015002}{DOI:\,10.1103/RevModPhys.89.015002}}.

\bibitem{ColesUncertaintyRel}
P.~J. Coles, R.~Colbeck, L.~Yu, and M.~Zwolak.
\newblock Uncertainty relations from simple entropic properties.
\newblock {\em Phys. Rev. Lett.}, 108:210405, May 2012.
\newblock
  \texttt{\href{http://dx.doi.org/10.1103/PhysRevLett.108.210405}{DOI:\,10.1103/PhysRevLett.108.210405}}.

\bibitem{ColesTripartite}
P.~J. Coles, L.~Yu, V.~Gheorghiu, and R.~B. Griffiths.
\newblock Information-theoretic treatment of tripartite systems and quantum
  channels.
\newblock {\em Phys. Rev. A}, 83:062338, Jun 2011.
\newblock
  \texttt{\href{http://dx.doi.org/10.1103/PhysRevA.83.062338}{DOI:\,10.1103/PhysRevA.83.062338}}.

\bibitem{curty2010}
M.~Curty, X.~Ma, H.-K. Lo, and N.~L\"utkenhaus.
\newblock Passive sources for the {B}ennett-{B}rassard 1984
  quantum-key-distribution protocol with practical signals.
\newblock {\em Phys. Rev. A}, 82:052325, Nov 2010.
\newblock
  \texttt{\href{http://dx.doi.org/10.1103/PhysRevA.82.052325}{DOI:\,10.1103/PhysRevA.82.052325}}.

\bibitem{dupuis}
F.~Dupuis, O.~Fawzi, and S.~Wehner.
\newblock Entanglement sampling and applications.
\newblock {\em IEEE Transactions on Information Theory}, 61(2):1093--1112, Feb
  2015.
\newblock
  \texttt{\href{http://dx.doi.org/10.1109/TIT.2014.2371464}{DOI:\,10.1109/TIT.2014.2371464}}.

\bibitem{frank}
R.~L. Frank and E.~H. Lieb.
\newblock Extended quantum conditional entropy and quantum uncertainty
  inequalities.
\newblock {\em Communications in Mathematical Physics}, 323(2):487--495, 2013.
\newblock
  \texttt{\href{http://dx.doi.org/10.1007/s00220-013-1775-1}{DOI:\,10.1007/s00220-013-1775-1}}.

\bibitem{furrer}
F.~Furrer, M.~Berta, M.~Tomamichel, V.~B. Scholz, and M.~Christandl.
\newblock Position-momentum uncertainty relations in the presence of quantum
  memory.
\newblock {\em Journal of Mathematical Physics}, 55(12):122205, 2014.
\newblock
  \texttt{\href{http://dx.doi.org/10.1063/1.4903989}{DOI:\,10.1063/1.4903989}}.

\bibitem{hall}
M.~J.~W. Hall.
\newblock Information exclusion principle for complementary observables.
\newblock {\em Phys. Rev. Lett.}, 74:3307--3311, Apr 1995.
\newblock
  \texttt{\href{http://dx.doi.org/10.1103/PhysRevLett.74.3307}{DOI:\,10.1103/PhysRevLett.74.3307}}.

\bibitem{heisenberg}
W.~Heisenberg.
\newblock \"{U}ber den anschaulichen {I}nhalt der quantentheoretischen
  {K}inematik und {M}echanik.
\newblock {\em Zeitschrift f\"{u}r Physik}, 43(3-4):172, 1927.
\newblock
  \texttt{\href{http://dx.doi.org/10.1007/BF01397280}{DOI:\,10.1007/BF01397280}}.

\bibitem{Helstrom}
C.~W. Helstrom.
\newblock {\em Quantum detection and estimation theory}.
\newblock Academic Press New York, 1976.

\bibitem{TernoQBS}
R.~Ionicioiu and D.~R. Terno.
\newblock Proposal for a quantum delayed-choice experiment.
\newblock {\em Phys. Rev. Lett.}, 107:230406, 2011.
\newblock
  \texttt{\href{http://dx.doi.org/10.1103/PhysRevLett.107.230406}{DOI:\,10.1103/PhysRevLett.107.230406}}.

\bibitem{QBSExp2}
F.~Kaiser, T.~Coudreau, P.~Milman, D.~B. Ostrowsky, and S.~Tanzilli.
\newblock Entanglement-enabled delayed-choice experiment.
\newblock {\em Science}, 338(6107):637, 2012.
\newblock
  \texttt{\href{http://dx.doi.org/10.1126/science.1226755}{DOI:\,10.1126/science.1226755}}.

\bibitem{kennard}
E.~H. {Kennard}.
\newblock {Zur Quantenmechanik einfacher Bewegungstypen}.
\newblock {\em Zeitschrift fur Physik}, 44:326, 1927.
\newblock
  \texttt{\href{http://dx.doi.org/10.1007/BF01391200}{DOI:\,10.1007/BF01391200}}.

\bibitem{Konig}
R.~Konig, R.~Renner, and C.~Schaffner.
\newblock The operational meaning of min- and max-entropy.
\newblock {\em IEEE Transactions on Information Theory}, 55(9):4337--4347, Sept
  2009.
\newblock
  \texttt{\href{http://dx.doi.org/10.1109/TIT.2009.2025545}{DOI:\,10.1109/TIT.2009.2025545}}.

\bibitem{korzekwa}
K.~Korzekwa, M.~Lostaglio, D.~Jennings, and T.~Rudolph.
\newblock Quantum and classical entropic uncertainty relations.
\newblock {\em Phys. Rev. A}, 89:042122, Apr 2014.
\newblock
  \texttt{\href{http://dx.doi.org/10.1103/PhysRevA.89.042122}{DOI:\,10.1103/PhysRevA.89.042122}}.

\bibitem{liu}
S.~Liu, L.-Z. Mu, and H.~Fan.
\newblock Entropic uncertainty relations for multiple measurements.
\newblock {\em Phys. Rev. A}, 91:042133, Apr 2015.
\newblock
  \texttt{\href{http://dx.doi.org/10.1103/PhysRevA.91.042133}{DOI:\,10.1103/PhysRevA.91.042133}}.

\bibitem{luo}
S.~L. Luo.
\newblock Quantum versus classical uncertainty.
\newblock {\em Theoretical and Mathematical Physics}, 143(2):681--688, 2005.
\newblock
  \texttt{\href{http://dx.doi.org/10.1007/s11232-005-0098-6}{DOI:\,10.1007/s11232-005-0098-6}}.

\bibitem{QBSExp3}
A.~Peruzzo, P.~Shadbolt, N.~Brunner, S.~Popescu, and J.~L.
  O{\textquoteright}Brien.
\newblock A quantum delayed-choice experiment.
\newblock {\em Science}, 338(6107):634, 2012.
\newblock
  \texttt{\href{http://dx.doi.org/10.1126/science.1226719}{DOI:\,10.1126/science.1226719}}.

\bibitem{renes}
J.~M. Renes and J.-C. Boileau.
\newblock Conjectured strong complementary information tradeoff.
\newblock {\em Phys. Rev. Lett.}, 103:020402, Jul 2009.
\newblock
  \texttt{\href{http://dx.doi.org/10.1103/PhysRevLett.103.020402}{DOI:\,10.1103/PhysRevLett.103.020402}}.

\bibitem{Robertson}
H.~P. Robertson.
\newblock The uncertainty principle.
\newblock {\em Phys. Rev.}, 34:163, 1929.
\newblock
  \texttt{\href{http://dx.doi.org/10.1103/PhysRev.34.163}{DOI:\,10.1103/PhysRev.34.163}}.

\bibitem{sanchez}
J.~S{\'a}nchez-Ruiz.
\newblock Improved bounds in the entropic uncertainty and certainty relations
  for complementary observables.
\newblock {\em Physics Letters A}, 201(2):125--131, 1995.
\newblock
  \texttt{\href{http://dx.doi.org/10.1016/0375-9601(95)00219-S}{DOI:\,10.1016/0375-9601(95)00219-S}}.

\bibitem{QBSExp1}
J.-S. Tang, Y.-L. Li, C.-F. Li, and G.-C. Guo.
\newblock Revisiting {B}ohr's principle of complementarity with a quantum
  device.
\newblock {\em Phys. Rev. A}, 88:014103, 2013.
\newblock
  \texttt{\href{http://dx.doi.org/10.1103/PhysRevA.88.014103}{DOI:\,10.1103/PhysRevA.88.014103}}.

\bibitem{Watrous}
J.~Watrous.
\newblock Lecture notes on theory of quantum information.
\newblock Available online:
  \url{https://cs.uwaterloo.ca/~watrous/CS766/LectureNotes/08.pdf} [cited
  29.10.2015].

\bibitem{Zurek}
W.~H. Zurek.
\newblock Decoherence, einselection, and the quantum origins of the classical.
\newblock {\em Rev. Mod. Phys.}, 75:715, 2003.
\newblock
  \texttt{\href{http://dx.doi.org/10.1103/RevModPhys.75.715}{DOI:\,10.1103/RevModPhys.75.715}}.

\end{thebibliography}
\onecolumngrid
\appendix
% !TEX root = main.tex

\section{The uncertainty game: definitions and basic derivations}
\label{sec:appendixA}
\subsection{Time evolution of the quantum circuit}
Following the quantum circuit of the uncertainty game in Fig.~\ref{fig:circuit} (in the main article), we derive the explicit form of the density matrices that Bob needs to distinguish in order to win the game. There are different classes of games depending on the parameter $d$ corresponding to the dimension of the Fourier transform or equivalently, the number of possible outcomes of Alice.
Bob prepares a state $\rho_B$ of dimension $d$ and sends it to Alice in register $B$. She holds another register $R$ in a state $\rho_R (\gamma) = \frac{1}{2} (\dya{0}+\dya{1} + \gamma^*\dyad{0}{1}+\gamma \dyad{1}{0})$, where $\gamma \in \mathbb{C}$ and $\abs{\gamma} \leq 1$. This $\gamma$ determines how coherent the register is. Specifically, in the later part of this appendix we show that we can restrict $\gamma$ to be real and $\gamma \in [0,1]$. Hence at the beginning (time $t_1$) the total state of the entire system is:
\begin{equation}
\rho_{RB}(\gamma, \rho_B) = \rho_R (\gamma) \otimes \rho_B =  \frac{1}{2} (\dya{0}_R+\dya{1}_R+ \gamma^*\dyad{0}{1}_R + \gamma \dyad{1}{0}_R) \otimes \rho_B \, .
\end{equation}
The state $ \rho_R (\gamma)$ determines the measurement basis in the following way: $\ket{0}$ corresponds to the measurement in the standard basis and $\ket{1}$ to the measurement in the Fourier basis (which is represented by applying the Fourier transformation to Bob's state and then measuring in the standard basis). Hence, the choice of the measurement basis can be represented by the controlled Fourier transform:
\begin{equation}
U = \dya{0}_R \otimes \mathbb{I}_B + \dya{1}_R \otimes F_B \, .
\end{equation}
We adopt the following convention for the Fourier transform: $F\ket{j} = \frac{1}{\sqrt{d}} \sum_{k=0}^{d-1}\omega^{jk}\ket{k}$ with $\omega = \exp\left(\frac{2 \pi i}{d}\right)$ being the d-th root of unity. After Alice applies the above unitary, the state at time $t_2$ is:
\begin{align}
\rho'_{RB}(\gamma, d, \rho_B)	&= U\rho_{RB}(\gamma, \rho_B) U^\dag = U(\rho_R(\gamma) \otimes \rho_B) U^\dag\\
					&= \frac{1}{2}(\dya{0}_R \otimes \rho_B + \gamma^* \dyad{0}{1}_R \otimes \rho_B F^\dag_B+\gamma \dyad{1}{0}_R F_B \rho_B+ \dya{1}_R \otimes F_B\rho_B F^\dag_B) \, .
\end{align}
Then Alice performs her measurement and the outcome is stored in the output register X. The total state after the measurement at time $t_3$ is:
\begin{equation}
\rho_{RX} (\gamma, d, \rho_B) = \sum_{x} \Tr_B[(\mathbb{I}_R \otimes \dya{x}_B)  \rho'_{RB}(\gamma, d, \rho_B)] \otimes \dya{x}_{X} \, .
\end{equation}
Hence, we see that the subnormalised post-measurement states of the basis register corresponding to Alice's measurement outcome $x$ are:
\begin{equation}
\begin{aligned}
\tilde{\rho}^x_R (\gamma, d, \rho_B) 	&= p_x(d, \rho_B) \rho^x_R(\gamma, d, \rho_B) = \Tr_B [ \left(\mathbb{I}_R \otimes \dya{x}_B \right) \rho'_{RB} ]  \\
						&=
\frac{1}{2}   \left( \begin{array}{cc}
\bra{x}\rho_B\ket{x} & \gamma^* \bra{x}\rho_B F^\dag\ket{x} \\
\gamma \bra{x} F \rho_B\ket{x} & \bra{x}F\rho_BF^\dag \ket{x}
\end{array} \right) \, ,
\label{eq:rho_xmixedcomplexgamma}
\end{aligned}
\end{equation}
where $p_x(d, \rho_B) = \Tr[\tilde{\rho}^x_R (\gamma, d, \rho_B)]$ is the probability that Alice observes outcome $x \in \{0,1, ..., d-1\}$. Note that $p_x$ does not depend on $\gamma$, which only appears in the off-diagonal elements of $\tilde{\rho}^x_R$. These subnormalised $\tilde{\rho}^x_R$'s are the states to which Bob has access and so his ability to predict Alice's measurement outcome $\ket{x}$ is determined by how well he can distinguish the quantum states $\{\rho^x_R\}$ occurring with probabilities $\{p_x\}$.
\subsection{Simplifying lemmas}
In the second part of this appendix we prove two lemmas, which allow us to restrict the coherence parameter $\gamma$ to real and positive numbers and the input state $\rho_B$ to pure states.
\begin{lemma}
In our problem, we can describe all the possible qualitatively different games just with $\gamma \in [0,1]$. That is, all games corresponding to $\gamma \in \mathbb{C}, \abs{\gamma} \leq 1$ are equivalent to some game with $\gamma \in [0,1]$.
\end{lemma}
\begin{proof}
Let $\gamma = \abs{\gamma} e^{i\theta}$. Then:
\begin{equation}
\begin{aligned}
\tilde{\rho}^x_R (\gamma, d, \rho_B) 	&=
\frac{1}{2}   \left( \begin{array}{cc}
\bra{x}\rho_B\ket{x} & \abs{\gamma} e^{-i\theta} \bra{x}\rho_B F^\dag\ket{x} \\
\abs{\gamma} e^{i\theta} \bra{x} F \rho_B\ket{x} & \bra{x}F\rho_BF^\dag \ket{x}
\end{array} \right) \, ,
\label{eq:rho_xmixedrealgamma}
\end{aligned}
\end{equation}
Let $V(\theta)$ denote the rotation matrix in the $xy$ plane of the Bloch sphere by angle $\theta$. That is:
\begin{equation}
\begin{aligned}
V(\theta) 	&= \left( \begin{array}{cc}
1 & 0 \\
0 & e^{i\theta}
\end{array} \right) \, .
\label{eq:rotmat}
\end{aligned}
\end{equation}
Then it can be easily verified that:
\begin{equation}
\tilde{\rho}^x_R (\gamma, d, \rho_B) = V(\theta) \tilde{\rho}^x_R (\abs{\gamma}, d, \rho_B) V^\dag(\theta) \, ,
\end{equation}
where $\abs{\gamma} \in [0,1]$. Hence all the output states $\tilde{\rho}^x_R (\gamma, d)$ up to a unitary rotation $V(\theta)$ are the same as the corresponding states $\tilde{\rho}^x_R (\abs{\gamma}, d)$. Clearly, rotating all the output states of register $R$ by a fixed angle $\theta$ does not affect their distinguishability. Hence, it is sufficient to consider real and positive $\gamma \in [0,1]$. 
\end{proof}
The probability of successfully discriminating states $\rho^x_R (\gamma, d, \rho_B)$, optimised over all measurements is~\cite{Watrous}:
\begin{equation}
\pguess(\gamma, d, \rho_B) = \max_{\{M_x\}}\sum_{x=0}^{d-1}p_x (d, \rho_B) \Tr[M_x \rho^x_R (\gamma, d, \rho_B)] = \max_{\{M_x\}}\sum_{x=0}^{d-1}\Tr[M_x \tilde{\rho}^x_R (\gamma, d, \rho_B)] \, ,
\label{eq:pguess}
\end{equation}
where $\{M_x\}$ is a POVM. Here, by $\pguess$ we denote the guessing probability optimised over all POVM's but for a specific input state $\rho_B$, while later we will use $\pguessmax$ to denote the guessing probability  $\pguess$ optimised over all inputs states of Bob. Both $\pguess$ and $\pguessmax$ are calculated for a specific game parameterised by $d \geq 2$ and for a specific $\gamma \in [0,1]$. Hence, we have $\displaystyle \pguessmax(\gamma,d) = \max_{\rho_B}\pguess(\gamma,d, \rho_B)$.
\begin{lemma}
To achieve $\pguessmax$ it is sufficient for Bob to consider pure input states. 
\end{lemma}
\begin{proof}
Firstly, let us consider the case when not only does Bob hold no quantum memory, but he also does not have any classical memory. Consider then a scenario in which Bob sends Alice a mixed state $\rho_B = \sum_{i} q_i \dya{\phi_i}$, where he is given freedom to choose the probabilities $\{q_i\}$. Then using Eq.~\eqref{eq:rho_xmixed}:
\begin{equation}
\begin{aligned}
\tilde{\rho}^x_R (\gamma, d, \rho_B) 	&= \sum_{i} q_i
\frac{1}{2}   \left( \begin{array}{cc}
\abs{\braket{x|\phi_i}}^2 & \gamma \braket{x|\phi_i} \braket{\phi_i|F^\dag|x} \\
\gamma \braket{\phi_i|x} \braket{x|F |\phi_i} & \abs{\braket{x|F|\phi_i}}^2
\end{array} \right) \\
						&= \sum_{i} q_i \tilde{\rho}^x_{R,i} (\gamma, d, \ket{\phi_i}) \, ,
\label{eq:rho_xmixedproof}
\end{aligned}
\end{equation}
where  $\tilde{\rho}^x_{R} (\gamma, d, \ket{\phi_i})$ denotes a post-measurement register state $\tilde{\rho}^x_R (\gamma, d, \rho_B)$  corresponding to Bob inputting a pure state $\rho_B = \dya{\phi_i}$. In this case the guessing probability from Eq.~\eqref{eq:pguess} becomes:
\begin{equation}
\begin{aligned}
\pguess(\gamma, d, \rho_B)  &= \max_{\{M_x\}}\sum_{x=0}^{d-1}\Tr\left[M_x \sum_{i} q_i \tilde{\rho}^x_{R} (\gamma, d, \ket{\phi_i})\right] \leq  \sum_{i} q_i \max_{\{M_x\}}\sum_{x=0}^{d-1}\Tr[M_x \tilde{\rho}^x_{R} (\gamma, d, \ket{\phi_i})]  \\
				&= \sum_{i} q_i \: \pguess(\gamma, d, \ket{\phi_i}) \leq \max_{i} \pguess(\gamma, d, \ket{\phi_i}) = \pguess(\gamma, d, \ket{\phi_m}) \, ,
\label{eq:pguessoptimisationproblem}
\end{aligned}
\end{equation}
where $\displaystyle \pguess(\gamma, d, \ket{\phi_i}) = \max_{\{M_x\}}\sum_{x=0}^{d-1}\Tr[M_x \tilde{\rho}^x_{R} (\gamma, d, \ket{\phi_i})]$ and by index $m$ we denote the largest of all $\pguess(\gamma, d, \ket{\phi_i})$ over all $i$'s. Hence it is optimal for Bob to prepare a state $\rho_B = \sum_{i} q_i \dya{\phi_i} = \dya{\phi_m}$ (so that $q_i = \delta_{i,m}$), such that $ \ket{\phi_m} \in \{\ket{\phi_i}\}$  and $\displaystyle \pguess(\gamma, d, \ket{\phi_m}) = \max_{i} \pguess(\gamma, d, \ket{\phi_i})$.
 
Now, if we allow Bob to have classical memory, he could then prepare a mixed state $\rho_B$ which is classically correlated to this memory. Then for each of the states $\rho^i_B$, corresponding to the state of the classical memory $\ket{i}_{M}$, we need to solve a separate optimisation problem given by Eq.~\eqref{eq:pguess}. Hence, if Bob prepares a state:
\begin{equation}
\rho_{BM} = \sum_{i} s_i \rho^i_B \otimes \dya{i}_{M}
\end{equation}
according to the probability distribution $\{s_i\}$, then the guessing probability will be a weighted average of the individual guessing probabilities corresponding to each of the states $\rho^i_B$, namely:
\begin{equation}
\pguess(\gamma, d, \rho_B) = \sum_i s_i \pguess(\gamma, d, \rho^i_B) \leq \pguess(\gamma, d, \rho^k_B) \, ,
\end{equation}
where $\rho^k_B$ is the input state that gives the highest guessing probability out of all the states $\{\rho^i_B\}$. Hence, classical memory does not allow us to achieve guessing probability higher than individual $\rho^k_B$, for which (as we have just seen) the guessing probability is upper bounded by its value corresponding to the optimal pure state $\ket{\phi_m}$ in the decomposition $\rho^k_B = \sum_{i} q_i \dya{\phi_i}$.
\end{proof}
Hence we will restrict our attention to scenarios in which Bob prepares a pure state $\ket{\phi}_B$. In this case the post-measurement states of the basis register are:
\begin{align}
\tilde{\rho}^x_R (\gamma, d, \ket{\phi}_B) 	&=
\frac{1}{2}   \left( \begin{array}{cc}
\abs{\braket{x|\phi}}^2 & \gamma \braket{x|\phi} \braket{\phi|F^\dag|x} \\
\gamma  \braket{x|F |\phi}\braket{\phi|x} & \abs{\braket{x|F|\phi}}^2
\end{array} \right) \, .
\label{eq:rho_x}
\end{align}

% !TEX root = main.tex

\section{Guessing probability for two-dimensional game ($d=2$)}
\label{sec:appendixB}

In this appendix we prove Theorem~\ref{pguessmaxd2}. That is, we derive the analytical formula for the maximum guessing probability as a function of $\gamma \in [0,1]$, for a game with two-dimensional Fourier transform (Hadamard transform) in our circuit and two possible outcomes. In this game the state $\rho_B$ that Bob prepares is a qubit. The two possible outcomes for Alice are: 0 and 1. We firstly restate this theorem below.

\begin{reptheorem}{pguessmaxd2}
The maximum guessing probability for a two-dimensional game ($d=2$), optimised over all input states $\rho_B$ is given by:
\begin{equation}
\pguessmax(\gamma, d=2) = \frac{1}{2} \left(1 + \frac{\sqrt{2+2\gamma^2}}{2}\right) \, .
\end{equation}
In particular, for $\gamma = 1$ one achieves perfect guessing, that is $\pguessmax(\gamma=1, d=2) = 1$.
\end{reptheorem}
\begin{proof}
The guessing probability is determined by how well Bob can distinguish states $\tilde{\rho}^0_R$ and $\tilde{\rho}^1_R$ defined in Eq.~\eqref{eq:rho_x} (for convenience we will omit writing out explicitly the dependence on $\gamma$ and $d$). The problem of distinguishing two states has been solved by Helstrom~\cite{Helstrom} and the guessing probability is:
\begin{equation}
p_{\text{guess}}= \frac{1}{2} (1 + \norm{G}_1) \, ,
\label{eq:Helstrom}
\end{equation}
where $G = \tilde{\rho}^0_R - \tilde{\rho}^1_R = p_0 \rho^0_R - p_1 \rho^1_R$ and $\norm{\cdot}_1$ denotes the trace-norm of the matrix. Firstly we note that for $d=2$, $F = F^\dag = H$. Secondly, since $\rho_B$ is a qubit, it is convenient to use the Bloch sphere representation:
\begin{equation}
\rho_B = \frac{1}{2}\left(\mathbb{I}+\sum_i c_i \sigma_i\right) \, ,
\end{equation}
with $c_x^2 + c_y^2 + c_z^2 \leq 1$. Although we have already shown in Appendix~\ref{sec:appendixA} that the optimal guessing probability $\pguessmax$ will be achieved for a pure input state $\rho_B$, here we are interested in all the qubit states that achieve this maximum guessing probability (under the assumption of Bob having no classical memory; if Bob had access to some classical memory, then any mixture of such optimal states correlated with this memory would also be an optimal state). Hence, in this appendix we again assume $\rho_B$ to be an arbitrary (possibly mixed) qubit state. Plugging the Bloch sphere representation of $\rho_B$ into Eq.~\eqref{eq:rho_xmixed}, we can first calculate $\tilde{\rho}^0_R$ and $\tilde{\rho}^1_R$ and then $G$:
\begin{align}
G = \frac{1}{2}  \left( \begin{array}{cc}
c_z & \frac{\gamma(1-i \cdot c_y)}{\sqrt{2}}  \\
\frac{\gamma(1+i \cdot c_y)}{\sqrt{2}} & c_x \end{array} \right) \, .
\end{align}
The eigenvalues of $G$ are:
\begin{equation}
\lambda = \frac{(c_x + c_z) \pm \sqrt{(c_x-c_z)^2 + \gamma^2(1+c_y^2)}}{4} \, .
\end{equation}
Now, let us consider two cases:
\begin{enumerate} [(a)]
\item $\lambda_1 \cdot \lambda_2 \geq 0$.

Then $\norm{G}_1^{\text{a}} = \abs{\lambda_1} + \abs{\lambda_2} = \abs{c_x+c_z}/2$ (the superscript ``a'' labels the case $\lambda_1 \cdot \lambda_2 \geq 0$). We are interested in the maximum possible value of $\norm{G}_1^{\text{a}}$ for a given $\gamma$. Hence we want to maximise the expression $\abs{c_x+c_z}$ subject to the constraint $c_x^2 + c_y^2+  c_z^2 \leq 1$. Clearly, this gives us $\abs{c_x+c_z} \leq \sqrt{2}$. and so $\norm{G}_1^{\text{a, max}} \leq \frac{\sqrt{2}}{2}$. In particular, this bound is tight for $c_y=0$ and $c_x = c_z = \pm \frac{1}{\sqrt{2}}$ (those states clearly satisfy the condition $\lambda_1 \cdot \lambda_2 \geq 0$). Hence, $\norm{G}_1^{\text{a, max}}= \frac{\sqrt{2}}{2}$.

\item $\lambda_1 \cdot \lambda_2 < 0$.

Then:
\begin{align}
\lambda_1 &= \frac{(c_x+c_z) + \sqrt{(c_x-c_z)^2 + 2\gamma^2(1+c_y^2)}}{4} > 0 \\
\lambda_2 &= \frac{(c_x+c_z) - \sqrt{(c_x-c_z)^2 + 2\gamma^2(1+c_y^2)}}{4} < 0 \, .
\end{align}
Hence in this case:
\begin{equation}
\norm{G}_1^{\text{b}} = \lambda_1 - \lambda_2 = \frac{\sqrt{(c_x-c_z)^2 + 2\gamma^2(1+c_y^2)}}{2} \, .
\end{equation}
Now we need to optimise this expression subject to the constraint $c_x^2 + c_y^2 + c_z ^2  \leq 1$. Let us use a substitution $a = \frac{c_x-c_z}{\sqrt{2}}$ $b = \frac{c_x+c_z}{\sqrt{2}}$. Then the constraint becomes: $a^2 + c_y^2 + b^2 \leq 1$ and the norm of $G$ is:
\begin{equation}
\norm{G}_1^{\text{b}} = \frac{\sqrt{2a^2 + 2\gamma^2(1+c_y^2)}}{2} \, .
\end{equation}
Clearly, since the term $c_y^2$ is scaled by the positive factor $2\gamma^2 \leq 2$, while $a^2$ is scaled by a factor of exactly 2, optimising this expression corresponds to setting $a^2$ to its maximum possible value which is 1 (so that $c_x=-c_z = \pm \frac{1}{\sqrt{2}}$). Then $c_y=b=0$ (one can easily verify that those values satisfy the condition of (b) $\lambda_1 \cdot \lambda_2 < 0$, for all $\gamma \in [0,1]$). This gives:
\begin{equation}
\norm{G}_1^{\text{b, max}} = \frac{\sqrt{2+2\gamma^2}}{2} \, ,
\end{equation}
\end{enumerate}
Clearly $\norm{G}_1^{\text{b, max}} \geq \norm{G}_1^{\text{a, max}}$ for all $\gamma \in [0,1]$ (the equality relation holds only for $\gamma=0$). Hence:
\begin{equation}
\norm{G}_1^{\text{max}} = \frac{\sqrt{2+2\gamma^2}}{2} \, ,
\end{equation}
Using $\norm{G}_1^{\text{max}}$, for every $\gamma$ we can now calculate the maximum value of the guessing probability:
\begin{equation}
\pguessmax(\gamma, d=2) = \frac{1}{2} (1 + \norm{G}_1^{\text{max}}) = \frac{1}{2} \left(1 + \frac{\sqrt{2+2\gamma^2}}{2}\right) \, .
\label{eq:pguessmaxqubit}
\end{equation}
We see also that for a fully coherent register with $\gamma = 1$, we obtain $\pguessmax  = 1$.
\end{proof}

In order to find the optimal states we need to consider 3 separate cases depending on the value of $\gamma$.

\begin{itemize}

\item $\gamma=0$. In this case $\norm{G}_1^{\text{max}}= \frac{\sqrt{2}}{2}$. This value occurs for two classes of states. One of them satisfies $a^2=1$ and $b=c_y=0$ which gives two solutions: $c_x=-c_z = \pm \frac{1}{\sqrt{2}}$. Hence we obtain two states: $(c_x, c_y, c_z) =\left(\frac{1}{\sqrt{2}}, 0,  -\frac{1}{\sqrt{2}}\right)$ and $(c_x, c_y, c_z) = \left(-\frac{1}{\sqrt{2}}, 0,  \frac{1}{\sqrt{2}}\right)$. The other class can be seen by noticing that $\norm{G}_1^{\text{max}}= \frac{\sqrt{2}}{2}=\norm{G}_1^{\text{a, max}}$ and so it can also be obtained from the case (a) for two states that achieve this value: $(c_x, c_y, c_z) =\left(\frac{1}{\sqrt{2}}, 0,  \frac{1}{\sqrt{2}}\right)$ and $(c_x, c_y, c_z) = \left(-\frac{1}{\sqrt{2}}, 0,  -\frac{1}{\sqrt{2}}\right)$.

\item $\gamma \in (0,1)$. Here we only have the class $a^2=1$ and $b=c_y=0$, that is the states: $(c_x, c_y, c_z) =\left(\frac{1}{\sqrt{2}}, 0,  -\frac{1}{\sqrt{2}}\right)$ and $(c_x, c_y, c_z) = \left(-\frac{1}{\sqrt{2}}, 0,  \frac{1}{\sqrt{2}}\right)$.

\item $\gamma=1$. Now $\norm{G}_1^{\text{b}} = \frac{\sqrt{2a^2 + 2(1+c_y^2)}}{2}$, and so this expression subject to the Bloch sphere normalisation is maximised by the pure states satisfying $a^2 + c_y^2 = 1$ and $b=0$. These are all pure states with $c_z = -c_x$ and $c_y = \pm \sqrt{1-2c_x^2}$. We can use angular parametrisation of those coefficients, in which case we can write this entire family of states as $(c_x, c_y, c_z) =\left(\sin(\theta), \pm \sqrt{\cos(2\theta)},  -\sin(\theta)\right)$ for all $\theta \in [-\frac{\pi}{4}, \frac{\pi}{4}]$. Geometrically, these states correspond to all pure states on the Bloch sphere that lie in the plane perpendicular to the Hadamard rotation axis and Hadamard transformation rotates them by $\pi$ rad to their orthogonal complement.
\end{itemize}
From Eq.~\eqref{eq:pguessmaxqubit} we see that the lowest value of $\pguessmax$ occurs for $\gamma = 0$ and it is $\pguessmax  = \frac{1}{2} \left(1 + \frac{1}{\sqrt{2}}\right)$. As the basis register state is becoming more pure by letting $\gamma$ grow, the $\pguessmax$ grows, until $\pguessmax  = 1$ for  $\gamma = 1$. We can also rephrase the guessing probability in terms of the purity of the basis register:
\begin{equation}
\Tr[\rho_R^2] = \frac{1}{4} \Tr\left[
\begin{pmatrix}
1& \gamma \\
\gamma & 1 
\end{pmatrix} 
\begin{pmatrix}
1& \gamma \\
\gamma & 1 
\end{pmatrix}
\right]= \frac{1}{4} \Tr\left[
\begin{pmatrix}
1+ \gamma^2& 2\gamma \\
2\gamma & 1+ \gamma^2 
\end{pmatrix} 
\right] =\frac{1+ \gamma^2}{2} \, .
\end{equation}
Hence:
\begin{equation}
\pguessmax(\gamma, d=2) = \frac{1}{2} \left(1 + \sqrt{\Tr[\rho_R^2]}\right) \, .
\end{equation}

% !TEX root = main.tex

\section{Guessing probability for the d-dimensional game}
\label{sec:appendixC}
We have already seen that in two dimensions utilising entanglement allows for guessing with probability equal to 1. In higher dimensions however, we show that this is not possible. This fact is expressed in Theorem~\ref{perfectguessinghighdim} in the main text. We restate and prove this theorem below.
\begin{reptheorem}{perfectguessinghighdim}
For $d$-dimensional games with any $d>2$ it is not possible to achieve perfect guessing, i.e.,
\begin{equation}
\pguessmax(\gamma, d>2) < 1 \, , \quad\quad  \forall \hspace{2pt}\gamma\,.
\end{equation}
\end{reptheorem}
\begin{proof}
We construct a proof by contradiction. Let us assume that there exists $d>2$ and $\gamma \in [0,1],$ such that $\pguessmax(\gamma, d) =1$. Since the states
$\tilde{\rho}^x_R (\gamma, d, \ket{\phi})$ are two-dimensional, it is only possible to perfectly distinguish at most 2 such states (if they are orthogonal). Hence, that means that to achieve $\pguessmax(\gamma, d) =1$ it is required that at least $d-2$ output states $\rho^x_R$ occur with probability zero. Hence, $\tilde{\rho}^x_R \neq 0$ for at most two values of $x$. Let us denote those two values of $x \in \{0,1,...,d-1\}$ for which it is possible that $\tilde{\rho}^x_R \neq 0$ by $x_0$ and $x_1$. We assume that those values are distinct so that $x_0 \neq x_1$. Specifically, let us assume that $\tilde{\rho}^{x_0}_R \neq 0$, while $\tilde{\rho}^{x_1}_R$ may or may not be equal to zero. Then let us define $\mathcal{P} = \{0,1,...,d-1\} \setminus \{x_0, x_1\}$. Therefore we require that $\tilde{\rho}^x_R = 0$ for all $x \in \mathcal{P}$. Thus we obtain the following two requirements:
\begin{enumerate}[1)]
\item $\braket{x|\phi} = 0$ for all $x \in \mathcal{P}$, \\
\item $\braket{x|F|\phi}=0$ for all $x \in \mathcal{P}$.\\
\end{enumerate}
The requirement 1) implies that the physical input state of Bob must be of the form:
\begin{equation}
\ket{\phi} = \alpha_0 \ket{x_0} + \alpha_1 \ket{x_1} \, ,
\label{eq:vanishmatrix}
\end{equation}
with 
\begin{equation}
\abs{\alpha_0}^2 + \abs{\alpha_1}^2 = 1 \, .
\label{eq:normalisation}
\end{equation}
In this framework, the scenario in which only $\tilde{\rho}^{x_0}_R \neq 0$ would require $\alpha_1 = 0$.
Now, note that:
\begin{equation}
F^\dag \ket{j} = \frac{1}{\sqrt{d}} \sum_{k=0}^{d-1} \omega^{-jk} \ket{k} \, ,
\end{equation}
where $\omega = \exp\left(\frac{2 \pi i}{d}\right)$ and so:
\begin{equation}
\bra{\phi}F^\dag \ket{x} =  \frac{1}{\sqrt{d}} (\alpha_0^* \omega^{-xx_0} + \alpha_1^* \omega^{-xx_1}) \, .
\end{equation}
Then 2) implies that:
\begin{equation}
\alpha_0^* + \alpha_1^* \omega^{x(x_0-x_1)} = 0 \, , \quad\quad  \forall x \in \mathcal{P} \, .
\label{eq:alpharel}
\end{equation}
Eq.~\eqref{eq:alpharel} together with Eq.~\eqref{eq:normalisation} require that $\alpha_0$ and $\alpha_1$ are of the form:
\begin{align}
\alpha_0 = \frac{1}{\sqrt{2}} e^{i\theta_0} \, ,
\label{eq:probamplfor2jsalpha0} \\
\alpha_1 = \frac{1}{\sqrt{2}} e^{i\theta_1} \, .
\label{eq:probamplfor2jsalpha1}
\end{align}
The above requirement shows that $\alpha_1$ cannot be zero, which in turn means that the scenario in which only $\tilde{\rho}^{x_0}_R \neq 0$ is not possible. Plugging the above forms of $\alpha$'s into Eq.~\eqref{eq:alpharel} and using the fact that $\omega$ is the d-th root of unity, we obtain the following requirement:
\begin{equation}
\theta_0 \equiv \theta_1 + \pi + 2\pi \left[\frac{x}{d} (x_1 - x_0)\right] (\bmod\; 2\pi) \, , \quad\quad  \forall x \in \mathcal{P} \, .
\label{eq:angleconstraint}
\end{equation}
Note that for $d=3$, this expression can be easily satisfied since in this case $\abs{\mathcal{P}}=1$, so e.g. $\theta_0 = \theta_1 + \pi + 2\pi \left[\frac{x_{\mathcal{P}}}{d} (x_1 - x_0)\right]$, where $x_{\mathcal{P}} \in \mathcal{P}$ satisfies Eq.~\eqref{eq:angleconstraint}. Hence the case $d=3$ needs to be analysed separately. For $d>3$ this equation can be satisfied if and only if:
\begin{equation}
\frac{x_1 - x_0}{d} \in \mathbb{Z} \, ,
\end{equation}
where $\mathbb{Z}$ denotes the set of integers. However, $x_0, x_1 \in \{0, d-1\}$ and $x_0 \neq x_1$. Therefore this equation cannot be satisfied. Hence, for $d>3$, it is not possible to have $\pguess(\gamma,d) = 1$. Now, let us consider the case $d=3$. Eq.~\eqref{eq:vanishmatrix} and Eq.~\eqref{eq:probamplfor2jsalpha0}-\eqref{eq:angleconstraint} imply that
\begin{equation}
\ket{\phi} = \frac{1}{\sqrt{2}}\left(\ket{x_1} - \omega^{x_{\mathcal{P}}(x_1-x_0)}\ket{x_0}\right) \, ,
\end{equation}
where we fix the global phase by setting $\theta_1 = 0$. Since $x_{\mathcal{P}}, x_0, x_1$ must be all different, there are 6 possible states  $\ket{\phi}$ corresponding to the above expression. Let $\ket{\psi_{kl}} = \frac{1}{\sqrt{2}}\left(\ket{l} - \omega^{x_{\mathcal{P}}(l-k)}\ket{k}\right)$. Then note that for every value of $x_{\mathcal{P}}$, the state $\ket{\phi} =\ket{\psi_{kl}}$ with $x_0 = k, x_1 = l$ and the state $\ket{\phi} = \ket{\psi_{lk}}$ with $x_0 = l, x_1 = k$ up to the global phase correspond to exactly the same state, since:
\begin{equation}
\ket{\psi_{kl}} = \frac{1}{\sqrt{2}}\left(\ket{l} - \omega^{x_{\mathcal{P}}(l-k)}\ket{k}\right) = -\omega^{x_{\mathcal{P}}(l-k)}\frac{1}{\sqrt{2}}\left(- \omega^{x_{\mathcal{P}}(k-l)}\ket{l} + \ket{k}\right) = -\omega^{x_{\mathcal{P}}(l-k)} \ket{\psi_{lk}} \, .
\end{equation}
Hence, we need only to consider 3 separate cases:
\begin{itemize}
\item For $x_{\mathcal{P}}=0, x_0=1, x_1= 2$, that is when $\tilde{\rho}^0_R = 0$, we have:
\begin{equation}
\ket{\phi} = \frac{1}{\sqrt{2}}\left(\ket{2} - \ket{1}\right) \, .
\end{equation}
Then:
\begin{equation}
F\ket{\phi} =i \frac{1}{\sqrt{2}}\left(\ket{2} - \ket{1}\right) = i \ket{\phi} \, .
\end{equation}
This means that if we define a matrix
\begin{equation}
\rho_c(\gamma) = \frac{1}{2} \left( \begin{array}{cc}
1& -i \gamma \\
i \gamma & 1
 \end{array} \right) \, ,
\end{equation}
then $\tilde{\rho}^0_R = 0, \tilde{\rho}^1_R = \abs{\braket{1|\phi}}^2 \rho_c(\gamma), \tilde{\rho}^2_R = \abs{\braket{2|\phi}}^2 \rho_c(\gamma)$. Hence, $\tilde{\rho}^1_R = \tilde{\rho}^2_R = \frac{1}{2} \rho_c(\gamma)$ and so we see that $\tilde{\rho}^1_R$ and $\tilde{\rho}^2_R$ correspond to the same state $\rho_c(\gamma)$ occurring with probability 0.5. This means that guessing probability in this case is 0.5 for all $\gamma \in [0,1]$.

\item For $x_{\mathcal{P}}=1, x_0=2, x_1=0$ with $\tilde{\rho}^1_R = 0$ the input state is:
\begin{equation}
\ket{\phi} = \frac{1}{\sqrt{2}}\left(\ket{0} - \omega^{-2}\ket{2}\right) = \frac{1}{\sqrt{2}}\left(\ket{0} - \omega\ket{2}\right) \, .
\end{equation}
Then:
\begin{equation}
F\ket{\phi} = \frac{1}{\sqrt{6}} (1-\omega) \left(\ket{0} - \omega^2 \ket{2}\right) \, .
\end{equation}
Hence, 
\begin{align}
\tilde{\rho}^0_R	&= \frac{1}{4} \left( \begin{array}{cc}
1& \gamma \frac{1}{\sqrt{3}}(1-\omega^*)  \\
\gamma \frac{1}{\sqrt{3}}(1-\omega) & 1
 \end{array} \right) \, , \\
\tilde{\rho}^1_R	&= 0 \, , \\
\tilde{\rho}^2_R	&= \frac{1}{4} \left( \begin{array}{cc}
1 & \gamma \frac{1}{\sqrt{3}}(1-\omega^*)\omega^*  \\
\gamma \frac{1}{\sqrt{3}}(1-\omega)\omega & 1
 \end{array} \right) \, .
 \end{align}
One can now show that $\Tr[\tilde{\rho}^0_R\tilde{\rho}^2_R] \neq 0$ for all $\gamma \in [0,1]$. Hence those states are not orthogonal and perfect guessing is not possible.
 
\item For $x_{\mathcal{P}}=2, x_0=0, x_1 = 1$, with $\tilde{\rho}^1_R = 0$ the input state is:
\begin{equation}
\ket{\phi} = \frac{1}{\sqrt{2}}\left(\ket{1} - \omega^{2}\ket{0}\right) \, .
\end{equation}
Then:
\begin{equation}
F\ket{\phi} = \frac{1}{\sqrt{6}} \left((1-\omega^2) \ket{0} + \sqrt{3}i \ket{1}\right) \, .
\end{equation}
Hence, 
\begin{align}
\tilde{\rho}^0_R	&= \frac{1}{4} \left( \begin{array}{cc}
1 & \gamma \frac{1}{\sqrt{3}}(1-\omega^*)  \\
\gamma \frac{1}{\sqrt{3}}(1-\omega) & 1
 \end{array} \right) \, , \\
\tilde{\rho}^1_R	&= \frac{1}{2} \rho_c(\gamma) \, , \\
\tilde{\rho}^2_R	&= 0 \, .
 \end{align}
Again $\Tr[\tilde{\rho}^0_R\tilde{\rho}^1_R] \neq 0$ for all $\gamma \in [0,1]$. Hence also in this case perfect guessing is not possible.
\end{itemize}
We have shown that perfect guessing in $d=3$ case is not possible either. Therefore we conclude that for all $d>2$ and for all $\gamma \in [0,1], \: \pguessmax(\gamma, d) <1$.
\end{proof}
The case $\gamma=0$ is a special case and can be solved analytically for all $d \geq 2$. 
\begin{proposition}  \thlabel{classical}
For $\gamma=0$ the maximal guessing probability is:
\begin{equation}
\pguessmax(\gamma = 0, d)  = \frac{1}{2}\left(1 + \frac{1}{\sqrt{d}}\right) \, ,
\end{equation}
and under assumption of Bob having no classical memory, it is achieved if and only if Bob's input state $\rho_B$ belongs to the following family of pure states:
\begin{equation}
\ket{\phi_{jl}} = c \left(\ket{j} + \omega^{jl}F^\dag \ket{l}\right) \, ,
\end{equation}
where $\omega = \exp\left(\frac{2 \pi i}{d}\right), j,l \in \{0,1,...,d-1\}$ and $c = \sqrt{\frac{\sqrt{d}}{2 \sqrt{d}+ 2}}$.
\end{proposition}
\begin{proof}
If one measures in the standard basis, the guessing probability for a fixed input state $\rho_B$ is:
\begin{equation}
p_{\text{guess}}^{\text{standard}}(d, \rho_B) = \max_{l} \Tr[\dya{l} \rho_B] \, .
\end{equation}
If one measures in the Fourier basis:
\begin{equation}
p_{\text{guess}}^{\text{Fourier}}(d, \rho_B) = \max_l \Tr[\dya{l} F\rho_B F^\dag] =\max_l \Tr[F^\dag\dya{l}F\rho_B] \, .
\end{equation}
Since each measurement occurs with probability 50\% and in the classical game the register $R$ only tells Bob which measurement basis was used, the guessing probability optimised over all input states of Bob is:
\begin{align}
\pguessmax(\gamma=0,d)	&=  \frac{1}{2} \max_{\rho_B} (p_{\text{guess}}^{\text{standard}}(d, \rho_B) + p_{\text{guess}}^{\text{Fourier}}(d, \rho_B)) =  \frac{1}{2} \max_{\rho_B} \max_{j,l} \Tr[(\dya{j} + F^\dag\dya{l} F) \rho_B] \\
						&= \frac{1}{2}\max_{j,l} \norm{\dya{j} + F^\dag\dya{l} F}_{\infty} \, ,
\end{align}
where $\norm{\cdot}_{\infty}$ denotes the infinity norm. The matrix whose infinity norm we need to find is a rank-2 matrix. Let $\pguess = \frac{1}{2}\norm{M}_{\infty}$ and $M = \dya{\alpha} + \dya{\beta}$ be a rank-2 matrix. The largest eigenvalue of such a matrix is $\norm{M}_{\infty} = \lambda_{\text{max}} = 1 + \abs{\braket{\alpha|\beta}}$. In our case: $\ket{\alpha} = \ket{j}$ and $\ket{\beta} = F^\dag\ket{l}$. This means that $\norm{M}_{\infty} = 1 + \frac{1}{\sqrt{d}}$ and so:
\begin{equation}
\pguessmax(\gamma=0,d) = \frac{1}{2}\left(1 + \frac{1}{\sqrt{d}}\right) \, .
\label{eq:pguessmaxgamma0}
\end{equation}
The eigenstate corresponding to this eigenvalue  $\lambda_{\text{max}}$ is:
\begin{equation}
\ket{\phi_{jl}} = c \left(\ket{j} + \omega^{jl}F^\dag \ket{l}\right) \, .
\end{equation}
Hence only the states of this form will give us the maximum guessing probability.
\end{proof}
We will now show that for a subclass of the states of this form Bob will be guessing always either $j$ or $l$, for all $\gamma \in [0,1]$ and all $d \geq 2$, since those 2 outcomes have much higher probabilities of occurrence $p_j(d, \ket{\phi_{jl}})$ and $p_l(d, \ket{\phi_{jl}})$ than all other outcomes (i.e. we will show that for input state $\ket{\phi_{jl}} = c \left(\ket{j} + \omega^{jl}F^\dag \ket{l}\right)$ such that $j \neq l$ the optimal strategy aims at distinguishing only the two states $\tilde{\rho}^j_R(\gamma,d, \ket{\phi_{jl}})$ and $\tilde{\rho}^l_R(\gamma,d, \ket{\phi_{jl}})$).
\begin{lemma} \thlabel{lookingon2} 
For all $d \geq 2$, for all $\gamma \in [0,1]$ and for all states $\ket{\phi_{jl}} = c \left(\ket{j} + \omega^{jl}F^\dag \ket{l}\right)$, such that $j,l \in \{0,1,...,d-1\}$ and $j \neq l$, the optimal guessing probability can be achieved by Bob if his measurement on the state of register $R$ is a POVM with only two occurring outcomes, that is the matrix elements of this POVM are: $M_j \neq 0$, $M_l \neq 0, M_k = 0$, for all $k \in \mathcal{P}$, where $\mathcal{P} = \{0,1,...,d-1\} \setminus \{j, l\}$.
\end{lemma}
\begin{proof}
The case $d=2$ is trivial, since then there are only two output states.

Now considering the general case, let $\lambda_{\text{min}}(\gamma, d, \ket{\phi_{jl}})$ denote the guessing probability corresponding to this restricted POVM. The ``min'' subscript indicates that this guessing probability is a lower bound on $\pguess(\gamma, d, \ket{\phi_{jl}})$, the guessing probability optimised over all POVMs. That is: $\lambda_{\text{min}}(\gamma, d, \ket{\phi_{jl}}) \leq \pguess(\gamma, d, \ket{\phi_{jl}})$. We then have:
\begin{equation}
\lambda_{\text{min}}(\gamma, d, \ket{\phi_{jl}}) = \max_{M_j, M_l}\Tr[M_j \tilde{\rho}^j_R(\gamma, d, \ket{\phi_{jl}})] + \Tr[M_l \tilde{\rho}^l_R(\gamma, d, \ket{\phi_{jl}})] \, ,
\end{equation}
Effectively this is again the problem of distinguishing 2 states solved by Helstrom~\cite{Helstrom}, the only difference is that this time $p_j(d, \ket{\phi_{jl}}) + p_l(d, \ket{\phi_{jl}}) \leq 1$. Hence
\begin{equation}
\lambda_{\text{min}}(\gamma, d, \ket{\phi_{jl}}) = \frac{1}{2} \left[\norm{G(\gamma, d, \ket{\phi_{jl}})}_1  + p_j(d, \ket{\phi_{jl}}) + p_l(d, \ket{\phi_{jl}})\right] \, ,
\end{equation}
where $G(\gamma,d, \ket{\phi_{jl}}) = \tilde{\rho}^j_R(\gamma, d, \ket{\phi_{jl}}) - \tilde{\rho}^l_R(\gamma, d, \ket{\phi_{jl}})$. Now we will show that this bound is tight, i.e. we will show that the above $\lambda_{\text{min}}(\gamma, d, \ket{\phi_{jl}})$ is in fact also an upper bound on $\pguess(\gamma, d, \ket{\phi_{jl}})$. For this purpose let us consider the dual program~\cite{Watrous} in which we consider all matrices 
\begin{equation}
Q(\gamma, d, \ket{\phi_{jl}}) \in \mathcal{Z}, \text{ where } \mathcal{Z} = \{Q \in \mathbb{C}^{2 \times 2}: Q = Q^\dag \land \forall k \in \{0,1,...,d-1\}, Q(\gamma, d, \ket{\phi_{jl}}) \geq \tilde{\rho}^k_R(\gamma, d, \ket{\phi_{jl}})\} \, .
\end{equation}
Then for each $Q \in \mathcal{Z}$ we define $\lambda_{\textnormal{max}}^Q(\gamma, d, \ket{\phi_{jl}}) = \Tr[Q(\gamma, d, \ket{\phi_{jl}})]$. From this it follows that $\pguess(\gamma, d, \ket{\phi_{jl}}) \leq \lambda_{\textnormal{max}}^Q(\gamma, d, \ket{\phi_{jl}})$ for all $Q \in \mathcal{Z}$~\cite{Watrous} and so $\lambda_{\textnormal{max}}^Q(\gamma, d, \ket{\phi_{jl}})$ is an upper bound on $\pguess(\gamma, d, \ket{\phi_{jl}})$. For simplicity, we will now omit writing explicitly the dependence on $\gamma, d$ and $\ket{\phi}$. Consider a hermitian matrix:
\begin{equation}
Q' = \frac{1}{2}(\tilde{\rho}^j_R + \tilde{\rho}^l_R + \abs{G}) \, .
\end{equation}
Then:
\begin{equation}
\Tr[Q'] =  \frac{1}{2}(p_j +p_l + \norm{G}_1) = \lambda_{\text{min}} \, .
\end{equation}
Now, if $Q'$ satisfies $Q' \geq \tilde{\rho}^k_R, \forall k$, then $Q' \in \mathcal{Z}$ and so $\Tr[Q'] = \lambda_{\textnormal{max}}^{Q'}$. And since then $\Tr[Q'] =\lambda_{\textnormal{min}} = \lambda_{\textnormal{max}}^{Q'}$, this means that $\Tr[Q'] = \pguess$. Hence, we will now prove that $\forall d \geq 3, \gamma \in [0,1]$ we have $Q' \in \mathcal{Z}$.

Consider
\begin{equation}
Q' - \tilde{\rho}^j_R = \frac{1}{2} ( - \tilde{\rho}^j_R +  \tilde{\rho}^l_R +  \abs{G}) = \frac{1}{2}(- G + \abs{G}) \, .
\end{equation}
Note that $\abs{G} \geq G$ and so $Q' - \tilde{\rho}^j_R \geq 0$. Hence $Q' \geq  \tilde{\rho}^j_R$. Analogously
\begin{equation}
Q' - \tilde{\rho}^l_R = \frac{1}{2}  (\tilde{\rho}^j_R -  \tilde{\rho}^l_R +  \abs{G}) =\frac{1}{2}(G + \abs{G}) \, .
\end{equation}
Clearly: $\abs{G} \geq -G$ and so $Q' - \tilde{\rho}^l_R \geq 0$. Hence $Q' \geq  \tilde{\rho}^l_R$.

Now we need to prove that $Q' \geq \tilde{\rho}^k_R, \forall k \in \mathcal{P}$ and for all $\gamma \in [0,1], d \geq 3$. In order to do that, we need to explicitly calculate all the output states of the register $R$. Those states are:
\begin{align}
\tilde{\rho}^j_R(\gamma, d, \ket{\phi_{jl}}) &= \frac{1}{2} \left( \begin{array}{cc}
A^2 & \gamma AB\omega^{-j^2}  \\
\gamma AB\omega^{j^2}  & B^2
 \end{array} \right) \, , \\
\tilde{\rho}^l_R(\gamma, d, \ket{\phi_{jl}}) &= \frac{1}{2} \left( \begin{array}{cc}
B^2 & \gamma AB \omega^{-l^2}  \\
\gamma AB \omega^{l^2} & A^2
 \end{array} \right) \, , \\
\tilde{\rho}^k_R(\gamma, d,\ket{\phi_{jl}}) &= \frac{B^2}{2} \left( \begin{array}{cc}
1 & \gamma \omega^{jl-jk-kl}  \\
\gamma \omega^{jk+kl-jl} & 1
 \end{array} \right) \, ,
 \end{align}
 where $A = c\left(1+\frac{1}{\sqrt{d}}\right), B = \frac{c}{\sqrt{d}}, k \in \mathcal{P}$. Then $Q' - \tilde{\rho}^k_R = \frac{1}{2} (\tilde{\rho}^j_R + \tilde{\rho}^l_R - 2\tilde{\rho}^k_R + \abs{G})$. Consider the operator:
\begin{equation}
D = \tilde{\rho}^j_R + \tilde{\rho}^l_R - 2\tilde{\rho}^k_R = \frac{1}{2} \left( \begin{array}{cc}
A^2-B^2 & \gamma B (A\omega^{-j^2} + A \omega^{-l^2} -2B \omega^{jl-jk-kl})  \\
\gamma B (A\omega^{j^2} + A \omega^{l^2} -2B \omega^{jk+kl-jl}) & A^2-B^2
 \end{array} \right) \, .
 \end{equation}
 We will now show that for all $k \in \mathcal{P}$ we have $D \geq 0$. Note that for $2 \times 2$ matrices, $D \geq 0$ if and only if $\Tr[D] \geq 0$ and $\textnormal{Det}(D) \geq 0$. Firstly, we see that $\Tr[D] = A^2-B^2 \geq 0, \forall d \geq 3$. Secondly, the determinant of $D$ is:
\begin{equation}
\begin{aligned}
\textnormal{Det}(D) 	&= \frac{1}{4}\left[(A^2-B^2)^2 - \gamma^2 B^2 \left(2A^2 + 4B^2 + 2A^2 \cos\left(\frac{2 \pi(j^2 - l^2)}{d}\right)\right.\right.\\
				&\left. \left. - 4AB\cos\left(\frac{2 \pi (j^2-jk-kl+jl)}{d}\right) - 4AB\cos\left(\frac{2 \pi (l^2-jk-kl+jl)}{d}\right)\right)\right] \, .
\end{aligned}
\end{equation}
Now we want to show that $\textnormal{Det}(D) \geq 0$ for all $j,l \in \{0,1,...,d-1\}, k \in \mathcal{P}, \gamma \in [0,1], d \geq 3$. From the above expression we see that $\textnormal{Det}(D)$ is monotonic in $\gamma \in [0,1]$. Clearly for $\gamma = 0, \textnormal{Det}(D) = \frac{1}{4}(A^2 - B^2)^2 \geq 0$. For $\gamma = 1$, we have:
\begin{equation}
\begin{aligned}
\textnormal{Det}(D) 	&= \frac{1}{4}\left[A^4 - 3B^4 - 4A^2B^2 - 2A^2B^2 \cos\left(\frac{2 \pi(j^2-l^2)}{d}\right)\right. \\
 				&\left. + 4AB^3 \cos\left(\frac{2 \pi (j^2 - jk - kl + jl)}{d}\right) +4AB^3 \cos\left(\frac{2 \pi (l^2-jk-kl+jl)}{d}\right)\right] \, .
\label{eq:detD}
\end{aligned}
\end{equation}
Note that $A = B(1+\sqrt{d})$. Thus we see that:
\begin{equation}
\begin{aligned}
\textnormal{Det}(D) 	&= \frac{B^4}{4}\left[(1+\sqrt{d})^4 - 3 - 4(1+\sqrt{d})^2 - 2(1+\sqrt{d})^2 \cos\left(\frac{2 \pi(j^2-l^2)}{d}\right)\right. \\
				&\left. + 4(1+\sqrt{d}) \cos\left(\frac{2 \pi (j^2 - jk - kl + jl)}{d}\right) +4(1+\sqrt{d}) \cos\left(\frac{2 \pi (l^2-jk-kl+jl)}{d}\right)\right]\\
				&\geq \frac{B^4}{4}\left[(1+\sqrt{d})^4 - 3 - 4(1+\sqrt{d})^2 - 2(1+\sqrt{d})^2 - 4(1+\sqrt{d})  - 4(1+\sqrt{d})\right] \\
				& =  \frac{B^4}{4}\left(d^2 + 4 d\sqrt{d} - 16 \sqrt{d} -16\right) \, .
\label{eq:detDbound}
\end{aligned}
\end{equation}
Let 
\begin{equation}
y(d) = \left(d^2 + 4 d\sqrt{d} - 16 \sqrt{d} -16\right) \, ,
\end{equation}
then $\textnormal{Det}(D) \geq \frac{B(d)^4}{4}y(d)$. Clearly $B(d) \geq 0, \forall d\geq 3$ and $y(d)  \geq 0, \forall d \geq 4$. Hence $\textnormal{Det}(D) \geq 0,  \forall d \geq 4$. For $d=3$ we use the exact expression from the first part of Eq.~\eqref{eq:detDbound} and we find that for all the cases $j \neq l$, $\forall k \in \mathcal{P}$, $\textnormal{Det}(D) \geq 0$. Hence  $\textnormal{Det}(D) \geq 0,  \forall d \geq 3$. Since both $\textnormal{Det}(D) \geq 0$ and $\Tr[D] \geq 0$, $D \geq 0$ and so $Q' \geq \tilde{\rho}^k_R, \forall k \in \{0,1,...,d-1\}$ and for all $\gamma \in [0,1], d \geq 3$. Therefore $Q' \in \mathcal{Z}$ and
\begin{equation}
\Tr[Q'] = \lambda_{\textnormal{max}}^Q(\gamma, d, \ket{\phi_{jl}}) = \lambda_{\textnormal{min}}(\gamma, d, \ket{\phi_{jl}})=\pguess(\gamma, d, \ket{\phi_{jl}}) \, .
\end{equation}
\end{proof}
Now, knowing that the strategy of distinguishing only the two most probable outcomes for the input state $\ket{\phi_{jl}} = c \left(\ket{j} + \omega^{jl}F^\dag \ket{l}\right)$, such that $j \neq l$ is actually an optimal strategy for those states, we can calculate the guessing probability for these states for all $d \geq 2$ and for all $\gamma \in [0,1]$:
\begin{equation}
\begin{aligned}
\pguess(\gamma, d, \ket{\phi_{jl}}) 	&= \frac{1}{2}(p_j +p_l + \norm{G}_1)\\ 
				&=\frac{1}{4(d+ \sqrt{d})}\left(2+2\sqrt{d}+d+\sqrt{d(2+\sqrt{d})^2 + 2\gamma^2 (1+\sqrt{d})^2 \left(1 - \cos\left(\frac{2 \pi (j^2 - l^2)}{d}\right)\right)}\right) \, .
\end{aligned}
\label{eq:pguessforlookingat2}
\end{equation}
Clearly for $\gamma = 0$ the above expression reduces to Eq.~\eqref{eq:pguessmaxgamma0}. That is $\pguess(\gamma=0, d, \ket{\phi_{jl}}) = \pguessmax(\gamma=0, d)$, since the states for which we have evaluated $\pguess(\gamma, d)$ above are the optimal states for $\gamma = 0$. Note that $A^2 = \pguessmax(\gamma=0, d)$ and so it is easy to see that for $\gamma=0$ the optimal measurement is:
\begin{equation}
M_j =  \left( \begin{array}{cc}
1 & 0  \\
0 & 0
 \end{array} \right) \, ,  \quad\quad
M_l = \left( \begin{array}{cc}
0 & 0  \\
0 & 1
 \end{array} \right) \, ,  \quad\quad
M_k = 0 \, , \quad\quad  \forall k \in \mathcal{P} \, .
\end{equation}
We can also see that for the game with $d=2$, the two cases $j=0, l=1$ and $j=1, l=0$ correspond to the two optimal states for all $\gamma \in [0,1]$. Hence, for these cases the above equation reduces to Eq.~\eqref{eq:pguessmaxqubit}.
\begin{lemma}\thlabel{monotonicity}
There exist states for which $\pguess(\gamma_1, d, \ket{\phi}) > \pguess(\gamma_2, d, \ket{\phi})>\pguessmax(\gamma=0, d),$ for $\gamma_1>\gamma_2>0, \forall d \geq 2$.
\end{lemma}
\begin{proof}
Consider all input states of the form $\ket{\phi_{jl}} = c \left(\ket{j} + \omega^{jl}F^\dag \ket{l}\right)$ such that $\frac{j^2-l^2}{d} \notin \mathbb{Z}$ and  $\forall d \geq 2$. Then firstly, $j \neq l$ and so the guessing probability corresponding to those states is given by Eq.~\eqref{eq:pguessforlookingat2} and secondly the coefficient in front of $\gamma^2$ is positive. Hence in these cases $\pguess(\gamma, d, \ket{\phi_{jl}})$ is monotonically increasing in $\gamma \in [0,1], \forall d \geq 2$. Hence, $\forall d \geq 2$, for all input states $\ket{\phi_{jl}} = c \left(\ket{j} + \omega^{jl}F^\dag \ket{l}\right)$ such that $\frac{j^2-l^2}{d} \notin \mathbb{Z}$ we have $\pguess(\gamma_1, d, \ket{\phi_{jl}}) > \pguess(\gamma_2, d, \ket{\phi_{jl}})>\pguessmax(\gamma=0, d),$ for $\gamma_1>\gamma_2>0$.
\end{proof}

Theorem~\ref{coherencevsnocoherence} follows directly from the above lemma by noting that $\pguessmax(\gamma, d) \geq \pguess(\gamma, d, \ket{\phi})$, for all $\gamma \in [0,1], d \geq 2$ and for all states $\ket{\phi}$.

One can also see that for the input states $\ket{\phi_{jl}} = c \left(\ket{j} + \omega^{jl}F^\dag \ket{l}\right)$ with $ j \neq l$ but with $\frac{j^2-l^2}{d} \in \mathbb{Z}$, Eq.~\eqref{eq:pguessforlookingat2} reduces to $\pguess(\gamma, d, \ket{\phi_{jl}}) = \frac{1}{2}\left(1 + \frac{1}{\sqrt{d}}\right) = \pguessmax(\gamma = 0, d)$. That is for those states $\pguess(\gamma, d, \ket{\phi_{jl}})$ stays constant in $\gamma$ for all $d$.
% !TEX root = main.tex

\section{Coherence and quantum correlations}
\label{sec:appendixD}

To give a deeper insight into the relation between the guessing probability and the coherence $\gamma$, we also look at the correlations between the registers $B$, $R$ and $P$ (the initial purification of $R$), at times $t_1$, $t_2$ and $t_3$ in Fig.~\ref{fig:circuit} (in the main article). Specifically, we focus on the two-dimensional game with optimal input states. We then quantify the arising correlations using min-entropy and the results are depicted in Fig.~\ref{fig:Hmin}. It needs to be noted that independently of the dimension of our game, Bob's requirements for perfect guessing are perfect classical correlations between $R$ and $X$, the classical register denoting the measurement outcome after Alice has performed her measurement on the system $B$ at time $t_3$ in Fig.~\ref{fig:circuit}. However, classical correlations are basis dependent and effectively the measurement of Alice involves two mutually unbiased bases. Hence it is impossible to have perfect guessing with just classically correlating the two systems before the measurement. From the perspective of the quantum circuit in Fig.~\ref{fig:circuit}, those perfect classical correlations that arise after the conditional Fourier transform will never be perfectly aligned with the measurement basis of Alice (standard basis). As a result, even if the system is classically perfectly correlated before the measurement, the correlations are no longer maximal after the measurement on $B$. For two-dimensional game, this can be seen in Fig.~\ref{fig:Hmin} where for $\gamma=0, \Hmin(B|R) = 0$, but $\Hmin(X|R)>0$. The advantage for Bob coming from the quantum coherence in register $R$ and the resulting quantum correlations is that for maximal entanglement (which is possible if $d=2$), independently of the basis in which the system $B$ has been measured, the outcomes of that measurement are maximally correlated with the state of the register $R$. Hence, if the two systems become maximally entangled ($\Hmin(B|R) = -1$ for $\gamma = 1$), then the post-measurement state becomes classically maximally correlated ($\Hmin(X|R) = 0$) enabling perfect guessing.
\begin{figure}[h]
  \centering
    \includegraphics[width=0.7\textwidth]{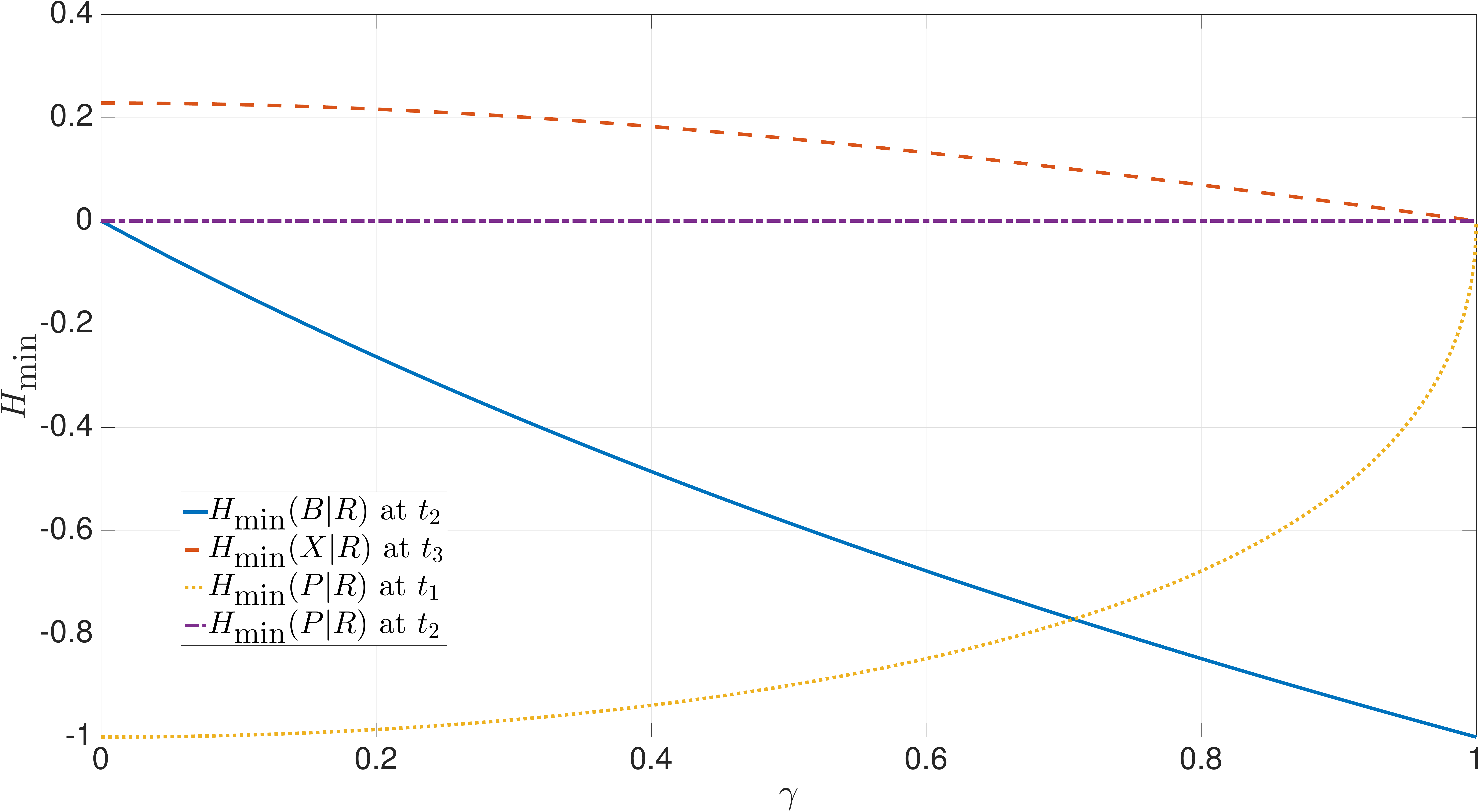}
    \caption{Conditional min-entropies as a function of $\gamma$ for the two-dimensional game ($d=2$) with Bob's input state $\ket{\phi_{01}} = c(\ket{0} + \ket{-})$ or $\ket{\phi_{10}} = c(\ket{1} + \ket{+})$. The blue solid line corresponds to the $\Hmin(B|R)$ at time $t_2$ in Fig.~\ref{fig:circuit}. The red dashed line shows $\Hmin(X|R)$ at time $t_3$ after Alice's measurement, where the state is averaged over all the outcomes, as Bob does not have access to the measurement result. The yellow dotted line corresponds to $\Hmin(P|R)$ at time $t_1$ and hence shows the initial quantum correlations between $R$ and its purification $P$. The correlations between those systems at time $t_2$ are illustrated by the purple dash-dotted flat line $\Hmin(P|R) =0$. By comparing the blue solid and red dashed lines, one can see that for $\gamma=1$ the increase of the conditional entropy between $\Hmin(B|R)$ and $\Hmin(X|R)$ due to the measurement on $B$ is the greatest possible, that is, it is equal to 1. The reason is that the measurement is the most destructive in this case, as it destroys all the quantum correlations of a maximally entangled state. On the other end of the spectrum, if $\gamma=0$, there are no quantum correlations between $B$ and $R$ present and so the measurement has a relatively small influence on the system. It only affects the classical correlations, which are not aligned with the standard basis in which the measurement performed by Alice takes place (the final measurement in the circuit in Fig.~\ref{fig:circuit}). Hence, in this case the increase of conditional entropy is small. Comparing the yellow dotted and blue solid lines we see that decreasing the amount of entanglement between $P$ and $R$ results in the increase in the amount of entanglement between $B$ and $R$ that can be generated using the controlled Fourier transform. Finally, from the flat purple dash-dotted line we see that independently of the coherence of $R$ and its initial correlations with $P$, the correlations between those two systems at time $t_2$ can be only classical. All the above entropies are derived in Appendix~\ref{sec:entropies}.}
		\label{fig:Hmin}
\end{figure}
% !TEX root = main.tex

\section{Conditional min-entropies for the two-dimensional game}
\label{sec:entropies}

The controlled Fourier transform in the circuit in Fig.~\ref{fig:circuit} (in the main article) results in (quantum) correlations between the two systems $B$ and $R$. These correlations are exploited by Bob in order to guess the measurement outcome on the state $\rho_B$. However, this measurement has a destructive effect on these correlations. Here we quantify this destructive effect of the measurement using min-entropy.
The conditional min-entropy will be calculated using the definition presented in~\cite{Konig}. Firstly let us define a correlation measure:
\begin{equation}
q_{\text{corr}}(B|R) = d \: \underset{\mathcal E}{\text{max}}\: F\left((\mathcal E_R \otimes \: \id_B)(\rho_{RB}),\dya{\Psi}_{RB} \right)^2 \, ,
\label{eq:qcorr}
\end{equation}
where $F$ is fidelity defined using the trace norm as $F(\rho,\sigma) = ||\sqrt{\rho}\sqrt{\sigma}||_1$ (when one of the states is pure, that is when $\sigma =  \dya{\Psi}$, the fidelity reduces to $F(\rho,\sigma) = \sqrt{\bra{\Psi}\rho \ket{\Psi}}$), $d$ is the dimension of subsystem $B$, $\mathcal E$ is a local operation described by a trace-preserving completely positive map and $\ket{\Psi}$ is a maximally entangled state (note that $q_{\text{corr}}(B|R)$ is independent of which maximally entangled state we use, since all such states are the same up to a unitary rotation on one of the qudits; this rotation can always be compensated on $\rho_{RB}$ by the corresponding rotation on system $R$ as part of the local operation $\mathcal E$). Then one can calculate the conditional min-entropy of a quantum-quantum (qq) state as $\Hmin(B|R)=-\log(q_{\text{corr}}(B|R))$. Note that for classical-quantum (cq) states, $q_{\text{corr}}(X|R)$ becomes the guessing probability $p_{\text{guess}}(X|R)$ (here X denotes the classical subsystem)~\cite{Konig}.

We are interested in the relation between the min-entropy $\Hmin(B|R)$ of a qq-state (the min-entropy of the input state $\rho_B$ before Alice's measurement, given access to $R$) and the min-entropy $\Hmin(X|R)$ of the cq-state after the measurement has been performed (the min-entropy of the classical outcome $X$ after Alice's measurement, given access to $\rho_R$). For that purpose we will investigate the tightness of the inequality derived in~\cite{Berta}:
\begin{equation}
\Hmin(X|R) \leq \Hmin(B|R) + \log(d) \, ,
\label{eq:Hmininequality}
\end{equation}
where $d$ is the dimension of the outcome space. This inequality tells us that for two-dimensional states, the increase of the conditional min-entropy due to the measurement cannot exceed 1.

For $d=2$ we will now calculate both of those entropies explicitly starting with $\Hmin(B|R)$. In our calculation let us pick one of the two states which give us the maximum guessing probability for all values of $\gamma$, namely $\ket{\phi_{10}}$ which in the Bloch sphere representation can be expressed as $(c_x, c_y, c_z) =\left(\frac{1}{\sqrt{2}}, 0,  -\frac{1}{\sqrt{2}}\right)$ [one can analogously show that the other state $\ket{\phi_{01}}$ or equivalently $(c_x, c_y, c_z) =\left(-\frac{1}{\sqrt{2}}, 0,  \frac{1}{\sqrt{2}}\right)$ will give exactly the same $\Hmin(B|R)$]. For this input state, the overall state ${\rho'_{RB}(\gamma, d=2, \ket{\phi})}$ before the measurement at time $t_2$ in Fig.~\ref{fig:circuit} is:
\begin{equation}
\begin{aligned}
\rho'_{RB}(\gamma, d=2, \ket{\phi}) 	&= \frac{1}{4}\left(\dya{0}_R \otimes (\mathbb{I}+ \frac{1}{\sqrt{2}} (\sigma_x -\sigma_z))+ \gamma[\dyad{0}{1}_R\otimes(H_B+ \frac{1}{\sqrt{2}} (\sigma_x - \sigma_z) H_B)\right. \\
					& \left. + \dyad{1}{0}_R\otimes(H_B+ \frac{1}{\sqrt{2}} H_B (\sigma_x -\sigma_z))]+ \dya{1}_R \otimes (\mathbb{I}+ \frac{1}{\sqrt{2}} (\sigma_z - \sigma_x))\right) \, .
\end{aligned}
\end{equation}
We can now diagonalise this state so that we obtain:
\begin{equation}
\rho'_{RB}(\gamma, d=2, \ket{\phi}) = \frac{1+\gamma}{2} \dya{\psi_1} + \frac{1-\gamma}{2} \dya{\psi_2} \, ,
\end{equation}
where the eigenstates written in their Schmidt bases are: 
\begin{align}
\ket{\psi_1} &= \frac{1}{\sqrt{2}}(\ket{0'}_R \ket{1}_B + \ket{1'}_R \ket{0}_B) \, , \\
\ket{\psi_2} &= \frac{1}{\sqrt{2}}(\ket{0''}_R \ket{1}_B + \ket{1''}_R \ket{0}_B) \, .
\end{align}
The Schmidt bases: $\{\ket{0'}, \ket{1'}\}$ and $\{\ket{0''}, \ket{1''}\}$ are given by:
\begin{align}
\ket{0'} &= \frac{1}{\sqrt{2}}\left(\frac{1}{\sqrt{2-\sqrt{2}}}\ket{0} - \frac{1}{\sqrt{2+\sqrt{2}}}\ket{1}\right) \, , \\
\ket{1'} &= \frac{1}{\sqrt{2}}\left(\frac{1}{\sqrt{2+\sqrt{2}}}\ket{0} + \frac{1}{\sqrt{2-\sqrt{2}}}\ket{1}\right) \, , \\
\ket{0''} &= \frac{1}{\sqrt{2}}\left(\frac{1}{\sqrt{2-\sqrt{2}}}\ket{0} + \frac{1}{\sqrt{2+\sqrt{2}}}\ket{1}\right) \, , \\
\ket{1''} &= \frac{1}{\sqrt{2}}\left(\frac{1}{\sqrt{2+\sqrt{2}}}\ket{0} - \frac{1}{\sqrt{2-\sqrt{2}}}\ket{1}\right) \, .
\end{align}
The states $\ket{\psi_1}$ and $\ket{\psi_2}$ are mutually orthogonal maximally entangled states. 
To calculate $\Hmin(B|R)$ we use the formulation of the min-entropy in terms of the semi-definite programmes, as expressed in~\cite{Konig}. The primal, as stated before, is $\Hmin(B|R)=-\log(q_{\text{corr}}(B|R))$ where $q_{\text{corr}}(B|R)$ is given in Eq.~\eqref{eq:qcorr}. The dual problem is:
\begin{equation}
\Hmin(B|R) = - \log \underset{\begin{subarray}{c}
  \sigma_R \geq 0 \\
   \sigma_R \otimes \mathbb{I}_B \geq \rho_{RB}
  \end{subarray}} {\min} \Tr(\sigma_R) \, .
  \label{eq:dual}
\end{equation}
For the primal programme, let us consider a local transformation $\mathcal E$ acting on subsystem $R$ which performs a rotation such that the state will now be diagonal in the basis that includes $\ket{\Psi}_{RB}$, with maximal probability in this mixture corresponding to the state $\ket{\Psi}_{RB}$. This feasible solution gives:
\begin{equation}
\underset{\mathcal E}{\text{max}}\: F\left((\mathcal E \otimes \: \id_B)(\rho'_{RB}),\dya{\Psi}_{RB}\right) \geq \sqrt{\frac{1+\gamma}{2}} \, .
\end{equation}
Hence:
\begin{equation}
q_{\text{corr}}(B|R) \geq 1+\gamma \, ,
\end{equation}
and so:
\begin{equation}
\Hmin(B|R) = -\log q_{\text{corr}}(B|R) \leq -\log (1+\gamma) \, .
\end{equation}
Similarly, for the dual programme, let us consider a matrix $\displaystyle \sigma_R = \left(\frac{1 + \gamma}{2}\right) \mathbb{I}_R \geq 0$. Then $\displaystyle \sigma_R \otimes \mathbb{I}_B =  \left(\frac{1+ \gamma}{2}\right) \mathbb{I}_{4 \times 4}$. Clearly $\displaystyle \sigma_R \otimes \mathbb{I}_B \geq \rho'_{RB}$, so that we obtain:
\begin{equation}
\Hmin(B|R) \geq - \log \Tr[\sigma_R] =  - \log(1+\gamma) \, .
\end{equation}
Combining the results from the primal and dual programmes allows us to conclude that ${\Hmin(B|R) =  -\log (1+\gamma)}$ for all $\gamma \in [0,1]$.

The min-entropy after the measurement is related to the guessing probability as $\Hmin(X|R)= -\log \pguess(X|R)$ and so it is:
\begin{equation}
\Hmin(X|R) = - \log \left(\frac{1}{2} \left(\frac{\sqrt{2+2\gamma^2}}{2} + 1\right)\right) = 1 - \log\left(\frac{\sqrt{2+2\gamma^2}}{2} + 1\right) \, .
\end{equation}
Hence:
\begin{equation}
\Hmin(X|R) - \Hmin(B|R) = 1 - \log\left(\frac{\sqrt{2+2 \gamma^2} + 2}{2(1+\gamma)}\right) \, .
\end{equation}
We then see that $\Hmin(X|R) - \Hmin(B|R)$ monotonically increases with $\gamma \in [0,1]$ until it reaches the value of one for $\gamma = 1$. Hence the inequality \eqref{eq:Hmininequality} is tight for $\gamma = 1$ which corresponds to the greatest possible increase of the conditional min-entropy during the measurement performed on a qubit (see Fig.~\ref{fig:Hmin}).

We also compute the min-entropy $\Hmin(P|R)$ to get some insight into the correlations between basis register $R$ and its purification $P$ as a function of $\gamma$. For that purpose, let us redefine the way we label the states of registers $R$ and $P$ with respect to the labelling and notation used in \crefrange{eq:subsystems}{eq:ngammameaning}. Specifically, let $\ket{\alpha},\ket{\beta}$ be now the two states of the entire register $P$ (joint states of all the environmental subsystems $E_i$ that are in $P$) corresponding to the states $\ket{0},\ket{1}$ of the register $R$ respectively. The real parameter $\gamma \in [0,1]$, that quantifies the amount of information that $P$ holds about $R$, satisfies now:
\begin{equation}
\braket{\alpha|\beta} = \gamma \, ,
\label{eq:alphabetagamma}
\end{equation}
so that the joint state of registers $R$ and $P$ can be written as:
\begin{equation}
\ket{\xi(\gamma)}_{RP} = \frac{1}{\sqrt{2}}(\ket{0}_R\ket{\alpha}_P + \ket{1}_R\ket{\beta}_P) \, .
\label{eq:jointinitialRP}
\end{equation}
Note that the state $\ket{\xi(\gamma)}_{RP}$ defined in Eq.~\eqref{eq:jointinitialRP} is pure. Then $\Hmin(P|R) = -\log(\Tr[\sqrt{\rho_R}])^2 = -\log(\Tr[\sqrt{\rho_P}])^2$. Note that $\Tr[\sqrt{\rho_R}] = \Tr[\sqrt{\rho_P}]$ is the sum of the Schmidt coefficients of the state $\ket{\xi(\gamma)}_{RP}$. The eigenvalues of $\rho_R(\gamma)$ defined in Eq.~\eqref{eq:rhoR} (with real and positive $\gamma$) are $\lambda_1 = \frac{1+\gamma}{2}$ and $\lambda_2 = \frac{1-\gamma}{2}$. Hence:
\begin{equation}
\Hmin(P|R) = -\log\left(\sqrt{ \frac{1+\gamma}{2}} + \sqrt{\frac{1-\gamma}{2}}\right)^2 = -\log(1+ \sqrt{1-\gamma^2}) \, .
\end{equation}

Similarly we calculate $\Hmin(P|R)$ after the conditional Fourier transform in Fig.~\ref{fig:circuit} has been applied, to quantify the effect of this operation on the correlations between $R$ and $P$. Firstly we need to calculate $\rho_{RP}$ at time $t_2$. That is, again following the circuit in Fig.~\ref{fig:circuit} but now including the purification $P$, the initial state at time $t_1$ is $\ket{\Phi(\gamma, d, \ket{\phi})}_{RPB} = \ket{\xi(\gamma)}_{RP}\otimes \ket{\phi}_B$. Then the state at time $t_2$ is $\ket{\Phi'(\gamma, d, \ket{\phi})}_{RPB} = U\ket{\Phi(\gamma, d, \ket{\phi})}_{RPB}$, where $U$ is given by:
\begin{equation}
U = \dya{0}_R \otimes \mathbb{I}_P \otimes \mathbb{I}_B + \dya{1}_R \otimes \mathbb{I}_P \otimes F_B \, .
\end{equation}
Hence:
\begin{equation}
\ket{\Phi'(\gamma, d, \ket{\phi})}_{RPB} = \frac{1}{\sqrt{2}}(\ket{0}_R \ket{\alpha}_P \ket{\phi}_B + \ket{1}_R \ket{\beta}_P F_B \ket{\phi}_B) \, ,
\end{equation}
We can now trace out $B$.
\begin{equation}
\begin{aligned}
\rho'_{RP}(\gamma, d, \ket{\phi}) &= \frac{1}{2}\left(\dya{0}_R \otimes \dya{\alpha}_P + \bra{\phi}F^\dag\ket{\phi} \dyad{0}{1}_R \otimes \dyad{\alpha}{\beta}_P \right. \\
																	&+ \left. \bra{\phi}F\ket{\phi} \dyad{1}{0}_R \otimes \dyad{\beta}{\alpha}_P + \dya{1}_R \otimes \dya{\beta}_P\right) \, .
\end{aligned}
\end{equation}
Now let us consider the two-dimensional game again with $\ket{\phi}_B$ being one of the two states that achieve $\pguessmax(\gamma, d=2)$  for all $\gamma \in [0,1]$ (these are the states $\ket{\phi} = \ket{\phi_{10}}$ and $\ket{\phi} = \ket{\phi_{01}}$). Then $\bra{\phi}F\ket{\phi}$ = 0, so the state on $R$ and $P$ at $t_2$ is:
\begin{equation}
\rho_{RP}(\gamma, d=2, \ket{\phi}) = \frac{1}{2}\left(\dya{0}_R \otimes \dya{\alpha}_P + \dya{1}_R \otimes \dya{\beta}_P\right) \, .
\end{equation}
To calculate $\Hmin(P|R)$ we again use the formulation of min-entropy in terms of the semi-definite programmes~\cite{Konig}. For the dual programme in Eq.~\eqref{eq:dual}, note that $\rho_{RP}$ has eigenvalues $\{\frac{1}{2}, \frac{1}{2}, 0, 0\}$. Hence $\displaystyle \sigma_R = \frac{\mathbb{I}_R}{2}$ clearly satisfies the constraints, as then $\displaystyle \sigma_R \otimes \mathbb{I}_P = \frac{\mathbb{I}_{4 \times 4}}{2}$ and so $\sigma_R \geq 0$ and $ \sigma_R \otimes \mathbb{I}_P \geq \rho_{RP}$. The corresponding solution is $\Hmin(P|R) \geq 0$. Similarly, in Eq.~\eqref{eq:qcorr}, let us consider $\mathcal E$ to be a quantum  channel acting on $R$ with Krauss operators $\{M_i\}$, where $M_0 = \dyad{\alpha}{0}$ and $M_0 = \dyad{\beta}{1}$.
Then:
\begin{equation}
\rho'_{RP} = (\mathcal E \otimes \: \id_P)(\rho_{RP}) = \frac{1}{2}\left(\dya{\alpha}_R \otimes \dya{\alpha}_P + \dya{\beta}_R \otimes \dya{\beta}_P\right) \, .
\end{equation}
Since $\braket{\alpha|\beta} = \gamma$, we have $\braket{\alpha^\bot|\beta} = e^{i\phi} \sqrt{1 - \gamma^2}$ for some phase $\phi$, where $\braket{\alpha|\alpha^\bot} = 0$. Now, let $\ket{\Psi}_{RP}$ be a maximally entangled state of the form $\ket{\Psi}_{RP} = \frac{1}{\sqrt{2}}(\ket{\alpha}_R\ket{\alpha}_P + e^{2i\phi}\ket{\alpha^\bot}_R\ket{\alpha^\bot}_P)$. Therefore:
\begin{equation}
\begin{aligned}
q_{\text{corr}}(P|R)	&= 2 F\left(\rho'_{RP}, \dya{\Psi}_{RP}\right)^2 \\
				&= \frac{1}{2}\left(\bra{\alpha}_R\bra{\alpha}_P + e^{-2i\phi}\bra{\alpha^\bot}_R\bra{\alpha^\bot}_P\right) \left(\dya{\alpha}_R \otimes \dya{\alpha}_P \right. \\
				&\left. + \dya{\beta}_R \otimes \dya{\beta}_P\right) \left(\ket{\alpha}_R\ket{\alpha}_P + e^{2i\phi}\ket{\alpha^\bot}_R\ket{\alpha^\bot}_P\right) \\
				&= \frac{1}{2}\left(1+ \abs{\braket{\alpha|\beta}}^4 + \abs{\braket{\alpha^\bot|\beta}}^4 + e^{2i\phi}(\braket{\alpha|\beta})^2 (\braket{\beta|\alpha^\bot})^2 + e^{-2i\phi}(\braket{\beta|\alpha})^2 (\braket{\alpha^\bot|\beta})^2 \right)\\
				&= \frac{1}{2}\left(1 + \gamma^4 + \left(1-\gamma^2\right)^2 + 2 \gamma^2  \left(1-\gamma^2\right)\right)\\
				&= 1 \, .
\end{aligned}
\end{equation}
Hence the corresponding solution is $\Hmin(P|R) \leq 0$. Therefore combining the results from the primal and dual programmes we conclude that $\Hmin(P|R) = 0$ for all $\gamma \in [0,1]$.
\end{document}